\documentclass[11pt]{article}
\usepackage[utf8]{inputenc}    
\usepackage{float}
\usepackage[style=ieee,sorting=none]{biblatex}
\bibliography{bibliography}
\usepackage{enumerate}
\usepackage{amsmath,braket,mathdots,mathtools,amssymb,xcolor,calligra,color}
\usepackage{amsthm} % JJ
\usepackage{bbm}  
\usepackage{cancel}
\usepackage{verbatim}
\usepackage[margin=1.14in]{geometry}
\usepackage{tikz}
\usetikzlibrary{spy}
\usetikzlibrary{patterns}
\usetikzlibrary{calc,arrows,positioning}
\usetikzlibrary{arrows,shadows}
\usetikzlibrary{decorations.pathmorphing}	% For Feynman 
\usetikzlibrary{decorations.markings}
\usepackage{pgfplots}
\usepackage{pgfplotstable}
\usepackage{graphicx,dblfloatfix}
\usepackage{hyperref}
\hypersetup{
    colorlinks=true,
    linkcolor=blue,
    filecolor=magenta,      
    urlcolor=cyan,
}
\newcommand{\be}{\begin{equation}}
\newcommand{\ee}{\end{equation}}
\newcommand{\bea}{\begin{eqnarray}}
\newcommand{\eea}{\end{eqnarray}}

\def\XXint#1#2#3{{\setbox0=\hbox{$#1{#2#3}{\int}$}
     \vcenter{\hbox{$#2#3$}}\kern-.5\wd0}}

\def\doi{http://dx.doi.org/}
\let\OLDthebibliography\thebibliography
\renewcommand\thebibliography[1]{
  \OLDthebibliography{#1}
  \setlength{\parskip}{0pt}
  \setlength{\itemsep}{0pt plus 0.3ex}
}

\newcommand{\cg}[1]{\bar{#1}}

\newcommand{\sign}{\,\text{sgn}\,}
\newcommand{\Tr}{\,\text{tr}\,}

\newcommand{\1}{\,\pmb{1}}
\newcommand{\D}[1]{\text{d}#1}
\newcommand{\intg}[1]{\lfloor #1\rfloor}
\newtheorem{theorem}{Theorem}

\newtheorem{lemma}{Lemma}

\begin{document}
\begin{center}
{\Large\bf Wavelet representation of hardcore bosons}
\end{center}
%%%%%%%%%%%%%%%%%%%%%%%%%%%%%%%%%%%%%%%
\begin{center}
Etienne Granet
\end{center}
\begin{center}
 Kadanoff Center for Theoretical Physics, University of Chicago, 5640 South Ellis Ave, Chicago, IL 60637, USA\\
% {\sf\small egranet@uchicago.edu}
\end{center}
\date{\today}

\section*{Abstract}
{\bf
We consider the 1D Tonks-Girardeau gas with a space-dependent potential out of equilibrium. We derive the exact dynamics of the system when divided into $n$ boxes and decomposed into energy eigenstates within each box. It is a representation of the wave function that is mixed between real space and momentum space, whose basis elements are plane waves localized in a box, motivating the word ``wavelet". In this representation we derive the emergence of generalized hydrodynamics in appropriate limits without assuming local relaxation. We emphasize in particular that a generalized hydrodynamic behaviour emerges in a high-momentum and short-time limit, besides the more common large-space and late-time limit, which is akin to a semi-classical expansion. In this limit, conserved charges \emph{do not} require a large number of particles to be described by generalized hydrodynamics. Besides, we show that this wavelet representation provides an efficient numerical algorithm for a complete description of out-of-equilibrium dynamics of hardcore bosons.
}

\section{Introduction}
Computing dynamics in many-body quantum systems is a difficult task even in systems that can be diagonalized exactly, since it involves going from a real space basis (to specify the initial state or to compute expectation values) to an energy eigenstate basis (to time-evolve the state). 
% Indeed, the basis that is well adapted to compute the time evolution of the state (the basis that diagonalizes the Hamiltonian) is very different from the basis that is well adapted to compute expectation values of observables or to specify the initial state (a basis of tensor products in real space). Going from one basis to the other typically involves exponentially many states, which is costly both in computational time and storage resource. 
In several areas of physics, \emph{hydrodynamics} is known to provide a good approximation of the dynamics of a many-body system in certain conditions \cite{bastianello2022introduction,ollitrault2008relativistic,cen1992hydrodynamic,marchetti2013hydrodynamics}. Hydrodynamics postulates that the system can be divided into ``fluid cells" that are relaxed to an equilibrium state, which in a microcanonical ensemble (or generalized microcanonical if there are several conserved charges) can be taken to be an energy eigenstate. At a larger scale, the system is thus described by equilibrium parameters (particle density, energy density, conserved charges, etc) that vary with space and time. Hydrodynamics is thus an approximate representation of a state that is mixed between real space (through the division into fluid cells) and energy eigenstates (through the description in terms of equilibrium parameters within each fluid cell). This mixed representation is key to its efficiency, as the real space resolution allows for a simple encoding of the initial state, while the energy eigenstate resolution allows for describing complex states locally with few parameters.

Recently, this hydrodynamic description has proven to be particularly efficient to describe 1D quantum integrable models \cite{castro2016emergent,bertini2016transport}, where due to the infinite number of conservation laws it is called ``generalized" hydrodynamics (GHD). A paradigmatic example is a gas of bosons in $\delta$ interaction, called the Lieb-Liniger model \cite{lieb1963exact}, which is known to provide an excellent description of cold atoms confined in a 1D geometry \cite{kinoshita2005local,haller2009realization,fabbri2011momentum,van2008yang,jacqmin2011sub,armijo2011mapping,schemmer2019generalized}. The limit of infinite coupling is particularly studied and is called the Tonks-Girardeau gas of hardcore bosons \cite{tonks1936complete,girardeau1960relationship}. But although exactly solvable with the Bethe ansatz, computing dynamics analytically or numerically in the Lieb-Liniger model remained a challenge for a long time. Contrary to the Bethe ansatz solution, this hydrodynamics description is amenable to efficient numerical implementation \cite{doyon2018soliton,moller2020introducing,bastianello2019generalized,bastianello2020thermalization,bulchandani2017solvable,moller2021extension}, yielding useful comparison material that was unavailable before. Besides its convenient practical use, it has also been reported to match remarkably well experimental \cite{schemmer2019generalized,kinoshita2006quantum,malvania2021generalized} and numerical data \cite{bertini2016transport,bulchandani2018bethe,doyon2017large,bastianello2019generalized}, see \cite{bouchoule2022generalized} for a review.

In this paper, inspired by the efficiency of hydrodynamics, we introduce an exact representation of hardcore bosons that is mixed between real space and momentum space. We explicitly implement the division of the system into fluid cells, and decompose a state in the basis given by the tensor product of energy eigenstates within each fluid cell. As the basis elements of this decomposition are small plane waves localized in each fluid cell, we call this decomposition a ``wavelet" representation. Contrary to a coarse-graining or a hydrodynamic description, we do \emph{not} remove any degree of freedom. Using the Girardeau mapping of hardcore bosons to free fermions \cite{girardeau1960relationship}, we derive the exact dynamics of the model in this wavelet representation. 

The output of this work is twofold. Firstly, (i) we show the emergence of (generalized) hydrodynamics in certain limits in absence of potentials, without assuming local relaxation of the gas within the fluid cells. The usual limit in which hydrodynamics holds is the \emph{Euler scale} limit, namely \emph{large space and large time} at fixed ratio space over time. In this scaling each fluid cell becomes infinitely large. But we show in particular that a hydrodynamic behaviour also emerges in a \emph{short-time, high-momentum} limit, which does not require any space rescaling. This is akin to a semi-classical expansion $\hbar\to 0$. In this limit, every particle with momentum $\alpha$ initialised in a fluid cell moves ``classically" from fluid cell to fluid cell with velocity $2\alpha$. In this scaling each fluid cell becomes infinitely small. Importantly, several observables (like particle density) \emph{do not} require a large number of particles to be described by GHD, which is compatible with observations that GHD works for very low numbers of particles \cite{schemmer2019generalized}. Despite this observation being very simple, it is appealing since much more physically applicable to cold-atom gases, and does not seem to have been observed elsewhere. To the best of our knowledge, previous mentions of the limit $\hbar\to 0$ in the context of GHD were always accompanied with the thermodynamic limit $N\sim 1/\hbar\to\infty$ \cite{ruggiero2020quantum,ruggiero2019conformal}.

Secondly, (ii) we show that this wavelet representation provides an efficient way of simulating hardcore bosons with inhomogeneous fields. Assuming that the initial state is a tensor product of energy eigenstates on each fluid cell, the state possesses a determinant structure at all times in the wavelet decomposition which can be computed efficiently. As said above, this representation is mixed between real space (through the division into fluid cells) and energy eigenstates space (through the decomposition into the local eigenstates in each fluid cell). This wavelet representation allows for a truncation of the higher momenta in each fluid cell, which is a physically meaningful truncation. We find that keeping a very small number of momenta in each fluid cell already gives very good results. Besides, this wavelet representation allows for a simple formula for the single-particle Green's function, which is usually difficult to compute \cite{wilson2020observation,rigol2005fermionization,papenbrock2003ground,minguzzi2002high,girardeau2001ground,pezer2007momentum}. As an illustration, we present simulations of quantum-Newton-cradle-like protocols for hardcore bosons with different confining potentials \cite{van2016separation,peotta2014quantum,atas2017exact}.\\

Some comments about GHD are in order, the hydrodynamic theory that describes the Lieb-Liniger model \cite{castro2016emergent,bertini2016transport}. Hydrodynamics is here ``generalized" in the sense that the Lieb-Liniger model being integrable, it has several other conserved charges beyond particle density, momentum and energy, encoded in a generating function $\rho(\lambda)$ called root density. The derivation of hydrodynamics in physical systems from first principles is generally a difficult task. The standard approach to argue for GHD is to divide the system into a collection of fluid cells assumed to be at equilibrium  \cite{doyon2020lecture,essler2022short}. Then, one computes the expectation value of the current of the conserved charges that parametrize the equilibrium state \cite{castro2016emergent,bertini2016transport,vu2019equations,urichuk2019spin,borsi2020current}. Popular toy models for these currents between fluid cells are bi-partitioning protocols \cite{antal1999transport,de2013nonequilibrium,collura2014quantum,viti2016inhomogeneous,perfetto2017ballistic,jin2021interplay}. This current leads to a change in the conserved quantities of the neighbouring fluid cells, yielding partial differential equations holding at a much larger scale than that of the fluid cells. Although intuitive, the drawback of this approach is that it entirely relies on the assumption of local equilibrium in each fluid cell, which is the difficult point to show. We also note that observations of large spatial correlations have  challenged the fluid cell picture recently \cite{doyon2022emergence}. 

The hardcore boson limit of the Lieb-Liniger model is particularly studied for ist mapping to free fermions \cite{girardeau1960relationship,tonks1936complete}. Finite couplings can then be studied perturbatively around free fermions \cite{cheon1999fermion,brand2005dynamic,cherny2006polarizability,granet2020systematic,granet2021systematic,granet2022duality}. For free fermions, any correlation function can be expressed in terms of the Wigner function $W(x,\lambda)$ through Wick's theorem. And this Wigner function straightforwardly satisfies a hydrodynamic-like equation holding at finite $x,t$ \cite{bettelheim2006orthogonality,bettelheim2008quantum,bettelheim2011universal,essler2022short, erdHos2004quantum}
\begin{equation}\label{ghdlike}
    \partial_t W(x,\lambda)+2\lambda \partial_x W(x,\lambda)=0\,.
\end{equation}
Moreover analogous continuity equations can be obtained for the Lieb-Liniger model perturbatively at large coupling around the free fermion limit \cite{bertini2022bbgky}. This equation is reminiscent to the GHD equation for the space-time dependent root density in the hardcore boson limit \cite{castro2016emergent,bertini2016transport}
\begin{equation}
    \partial_t \rho_{x,t}(\lambda)+2\lambda \partial_x \rho_{x,t}(\lambda)=0\,.
\end{equation}
Besides, the integration of $W(x,\lambda)$ over $x$ in an equilibrium state yields the root density $\rho(\lambda)$ of this equilibrium state. But to go from $W(x,\lambda;t)$ satisfying \eqref{ghdlike} to an equation on a local root density $\rho_{x,t}(\lambda)$ is more subtle and is where lies hydrodynamics. It would require again a coarse-graining, a separation of scale and an assumption of local relaxation \cite{essler2022short}, and would only hold in certain limits. For finite $x,t$, although \eqref{ghdlike} holds true, the expression for some observables at equilibrium in terms of $\rho(\lambda)$ would \emph{not} be given by the expression in terms of $W(y,\mu;t)$ obtained with Wick's theorem. This expression would have ``unwanted" terms that vanish only in certain limits.\\

The paper is organized as follows. In Section \ref{problemsetting} we introduce the model, the out-of-equilibrium setup and the wavelet basis. In Section \ref{exact} we derive the representation of the state of the system out-of-equilibrium in this wavelet basis. In Section \ref{expe} we explain how to compute expectation values in this representation. In Section \ref{hydrosec} we define the hydrodynamic limit that we consider, and derive a hydrodynamic behaviour in our setup. 
 In Section \ref{algosec} we present a numerical algorithm based on this wavelet for simulating the dynamics of hardcore bosons with inhomogeneous potentials.

\section{Problem setting \label{problemsetting}}
\subsection{Lieb-Liniger model and hardcore bosons\label{hardcore}}
We consider the Lieb-Liniger model on a ring of size $L$
\begin{equation}
H=\int_0^L \left[ -\phi^\dagger(x) \partial_x^2 \phi (x)+c\phi^\dagger(x)\phi^\dagger(x) \phi(x)\phi(x)\right]\D{x}\,,
\end{equation}
with a Bose field $\phi(x)$ and periodic boundary conditions. We will be interested in the limit $c\to\infty$ where the bosons become hardcore and map to free fermions \cite{girardeau1960relationship}. The normalized eigenfunctions in this limit $c\to\infty$ are parametrized by a set of distinct quantized momenta $\pmb{\lambda}=\{\lambda_1,...,\lambda_N\}\subset K$ with if $N$ is even \cite{korepin1997quantum}
\begin{equation}\label{kbar}
   K=\left\{\frac{2\pi (m+\tfrac{1}{2})}{L},\quad m\in\mathbb{Z}\right\}\,,
\end{equation}
and if $N$ is odd
\begin{equation}
   K=\left\{\frac{2\pi m}{L},\quad m\in\mathbb{Z}\right\}\,.
\end{equation}
For simplicity and without loss of generality, we will restrict to $N$ even for the rest of the paper, and take $K$ as defined in \eqref{kbar}. The explicit expression of the normalized eigenstates is 
\begin{equation}
|\pmb{\lambda}\rangle=\int_0^L...\int_{0}^L\D{x_1}...\D{x_N} \chi(x_1,...,x_N)\phi^\dagger(x_1)...\phi^\dagger(x_N)|0\rangle\,,
\end{equation}
with when $\lambda_1<...<\lambda_N$
\begin{equation}\label{chi}
\chi(x_1,...,x_N)=\frac{1}{N!L^{N/2}}\sum_{\sigma\in\mathfrak{S}_N}(-1)^\sigma \exp\left(i\sum_{j=1}^N x_j\lambda_{\sigma j} \right) \prod_{i<j}\sign(x_i-x_j)\,.
\end{equation}
Their energy is
\begin{equation}
    E(\pmb{\lambda})=\sum_{j=1}^N \lambda_j^2\,.
\end{equation}
% We note that we have chosen a convention such that the eigenstates are antisymmetric functions of the $\lambda_i$'s.

\subsection{Wavelet representation}
We divide the entire system $[0,L]$ into $n$ boxes $[(b-1) \tfrac{L}{n},b\tfrac{L}{n}]$ of equal size $L/n$, for $b=1,...,n$, as depicted in Fig \ref{kadanoff}. We can define in each box $b$ the same Hamiltonian $H_0$ restricted to this box, and set anti-periodic boundary conditions for even number of particles, and periodic boundary conditions for odd number of particles. The eigenstates of this Hamiltonian are parametrized by sets of distinct momenta $\pmb{\alpha^{(b)}}\subset K_{\rm box}$ with
\begin{equation}\label{k}
  K_{\rm box}= \left\{\frac{2\pi nm}{L},\quad m\in\mathbb{Z}\right\}\,,
\end{equation}
irrespectively of the parity of the number of particles, and have identical representations as in \eqref{chi}. The particular choice of boundary conditions is purely for technical convenience, so that elements of $K_{\rm box}$ can never be equal to elements of $K$ in \eqref{kbar}, and will not influence the physics. The important point is that they form a basis of the Hilbert space in box $b$, and that local observables can be expressed simply in terms of them. We note the factor $n$ compared to \eqref{kbar}, coming from the change of system size. 

As these states form a basis of the Hilbert space of wave functions on box $b$, a basis of the entire Hilbert space on $[0,L]$ is obtained with tensor products of theses states in each box. Namely, given $\pmb{\alpha^{(b)}}\subset K_{\rm box}$ for $b=1,...,n$ we introduce
\begin{equation}
|\pmb{\cg{\alpha}}\rangle=|\pmb{\alpha^{(1)}}\rangle \otimes...\otimes |\pmb{\alpha^{(n)}}\rangle\,,
\end{equation}
with by definition $\pmb{\cg{\alpha}}\subset K_\otimes$ where
\begin{equation}
    K_\otimes= \left\{\left(\frac{2\pi nm}{L},b\right),\quad m\in\mathbb{Z}\,,\quad b\in\{1,...,n\}\right\}\,.
\end{equation}
We will denote $\cg{\alpha}\in K_\otimes$ the couple $\cg{\alpha}=(\frac{2\pi nm}{L},b)$, with $\alpha=\frac{2\pi nm}{L}\in K_{\rm box}$ the momentum and $b(\cg{\alpha})=b\in\{1,...,n\}$ the box index. We will call $|\pmb{\cg{\alpha}}\rangle$ a \emph{wavelet} state. This term is motivated by the fact that $\pmb{\cg{\alpha}}$ is a collection of ``wavelets", i.e.\ small plane waves localized in each box. Explicitly, the wave function of the wavelet state $|\pmb{\cg{\alpha}}\rangle$ is
\begin{equation}\label{chiwavelet}
\begin{aligned}
&\chi(x_1,...,x_N)=\frac{n^{N/2}}{L^{N/2}\prod_{b=1}^nN_b!}\prod_{b=1}^n\Bigg[\sum_{\sigma\in\mathfrak{S}_{N_b}} (-1)^\sigma \exp\left(i\sum_{j=1}^{N_b} x_{N_1+...+N_{b-1}+j}\alpha^{(b)}_{\sigma j} \right)\\
&\prod_{i<j}\sign(x_{N_1+...+N_{b-1}+i}-x_{N_1+...+N_{b-1}+j})\prod_{j=1}^{N_b}\1_{\tfrac{b-1}{n}L<x_{N_1+...+N_{b-1}+j}<\tfrac{b}{n}L}\Bigg]\,,
\end{aligned}
\end{equation}
with $N_b$ denoting the number of particles in each $\pmb{\alpha^{(b)}}$, with $N_1+...+N_n=N$.

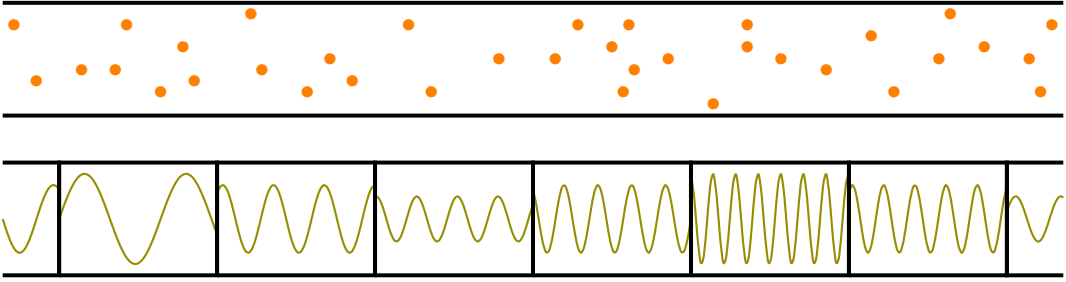
\begin{figure}[H]
\begin{center}
\begin{tikzpicture}[scale=1.5]

\draw[black, line width=1.5pt] (-0.5,1)--(6*1.4+0.5,1);
\draw[black, line width=1.5pt] (-0.5,0)--(6*1.4+0.5,0);

\draw node[orange] at (-0.2,0.3) { $\bullet$};
\draw node[orange] at (-0.4,0.8) { $\bullet$};

\draw node[orange] at (0.2,0.4) { $\bullet$};
\draw node[orange] at (0.6,0.8) { $\bullet$};
\draw node[orange] at (1.1,0.6) { $\bullet$};
\draw node[orange] at (0.9,0.2) { $\bullet$};
\draw node[orange] at (1.2,0.3) { $\bullet$};
\draw node[orange] at (0.5,0.4) { $\bullet$};

\draw node[orange] at (1.4+0.4,0.4) { $\bullet$};
\draw node[orange] at (1.4+0.8,0.2) { $\bullet$};
\draw node[orange] at (1.4+1.2,0.3) { $\bullet$};
\draw node[orange] at (1.4+1.0,0.5) { $\bullet$};
\draw node[orange] at (1.4+0.3,0.9) { $\bullet$};

\draw node[orange] at (2*1.4+1.1,0.5) { $\bullet$};
\draw node[orange] at (2*1.4+0.5,0.2) { $\bullet$};
\draw node[orange] at (2*1.4+0.3,0.8) { $\bullet$};

\draw node[orange] at (3*1.4+0.85,0.8) { $\bullet$};
\draw node[orange] at (3*1.4+1.2,0.5) { $\bullet$};
\draw node[orange] at (3*1.4+0.9,0.4) { $\bullet$};
\draw node[orange] at (3*1.4+0.7,0.6) { $\bullet$};
\draw node[orange] at (3*1.4+0.4,0.8) { $\bullet$};
\draw node[orange] at (3*1.4+0.8,0.2) { $\bullet$};
\draw node[orange] at (3*1.4+0.2,0.5) { $\bullet$};

\draw node[orange] at (4*1.4+0.2,0.1) { $\bullet$};
\draw node[orange] at (4*1.4+0.8,0.5) { $\bullet$};
\draw node[orange] at (4*1.4+1.2,0.4) { $\bullet$};
\draw node[orange] at (4*1.4+0.5,0.6) { $\bullet$};
\draw node[orange] at (4*1.4+0.5,0.8) { $\bullet$};

\draw node[orange] at (5*1.4+1.2,0.6) { $\bullet$};
\draw node[orange] at (5*1.4+0.4,0.2) { $\bullet$};
\draw node[orange] at (5*1.4+0.9,0.9) { $\bullet$};
\draw node[orange] at (5*1.4+0.2,0.7) { $\bullet$};
\draw node[orange] at (5*1.4+0.8,0.5) { $\bullet$};

\draw node[orange] at (6*1.4+0.2,0.5) { $\bullet$};
\draw node[orange] at (6*1.4+0.4,0.8) { $\bullet$};
\draw node[orange] at (6*1.4+0.3,0.2) { $\bullet$};

\end{tikzpicture}

\vspace{0.5cm}

\begin{tikzpicture}[scale=1.5]

\draw[scale=1, domain=0.9:1.4,  samples=50,smooth, variable=\x, olive, thick] plot ({-1.4+\x}, {0.5+0.3*sin(600*\x)});
\draw[scale=1, domain=0:1.4, samples=50,smooth, variable=\x, olive, thick] plot ({\x}, {0.5+0.4*sin(400*\x)});
\draw[scale=1, domain=0:1.4,  samples=50,smooth, variable=\x, olive, thick] plot ({1.4+\x}, {0.5+0.3*sin(50+800*\x)});
\draw[scale=1, domain=0:1.4,  samples=50,smooth, variable=\x, olive, thick] plot ({2*1.4+\x}, {0.5+0.2*sin(80+1000*\x)});
\draw[scale=1, domain=0:1.4,  samples=50,smooth, variable=\x, olive, thick] plot ({3*1.4+\x}, {0.5+0.3*sin(120+1200*\x)});
\draw[scale=1, domain=0:1.4, samples=50, smooth, variable=\x, olive, thick] plot ({4*1.4+\x}, {0.5+0.4*sin(100+1800*\x)});
\draw[scale=1, domain=0:1.4,  samples=50,smooth, variable=\x, olive, thick] plot ({5*1.4+\x}, {0.5+0.3*sin(50+1300*\x)});
\draw[scale=1, domain=0:0.5,  samples=50,smooth, variable=\x, olive, thick] plot ({6*1.4+\x}, {0.5+0.2*sin(20+900*\x)});

\draw[black, line width=1.5pt] (-0.5,0)--(0.,0)--(0.,1)--(-0.5,1);
\draw[black, line width=1.5pt] (0,0)--(1.4,0)--(1.4,1)--(0,1)--(0,0);
\draw[black, line width=1.5pt] (1.4,0)--(2*1.4,0)--(2*1.4,1)--(1.4,1)--(1.4,0);
\draw[black, line width=1.5pt] (2*1.4,0)--(3*1.4,0)--(3*1.4,1)--(2*1.4,1)--(2*1.4,0);
\draw[black, line width=1.5pt] (3*1.4,0)--(4*1.4,0)--(4*1.4,1)--(3*1.4,1)--(3*1.4,0);
\draw[black, line width=1.5pt] (4*1.4,0)--(5*1.4,0)--(5*1.4,1)--(4*1.4,1)--(4*1.4,0);
\draw[black, line width=1.5pt] (5*1.4,0)--(6*1.4,0)--(6*1.4,1)--(5*1.4,1)--(5*1.4,0);
\draw[black, line width=1.5pt] (6*1.4+0.5,1)--(6*1.4,1)--(6*1.4,0)--(6*1.4+0.5,0);

\end{tikzpicture}
\caption{Schematic picture of the wavelet decomposition. The original system (top) is divided into boxes (bottom) and described by plane waves with different amplitudes and frequencies within each box.} 
\label {kadanoff}
\end{center}
\end {figure}

\subsection{Inhomogeneities and out-of-equilibrium protocol}
Given the division into boxes defined above, we consider the following Hamiltonian
\begin{equation}
\begin{aligned}
    H(\{v_b\})&=H_0+V\,,\qquad \text{with }V=\sum_{b=1}^n v_b \int_{(b-1) \tfrac{L}{n}}^{b \tfrac{L}{n}} \phi^\dagger(x)\phi(x)\D{x}\,,
\end{aligned}
\end{equation}
where $v_b$ are potentials that are uniform in each box but that can differ from box to box. As it does not complicate the discussion, we will consider a  general setting where the potentials are given an arbitrary time dependence $v_b(t)$. For an initial wave function $|\Psi(0)\rangle$, we propose to study the dynamics induced by this Hamiltonian. Namely, we define
\begin{equation}\label{protocol}
    |\Psi(t)\rangle=U(t)|\Psi(0)\rangle\,,
\end{equation}
where the evolution operator is given by the time-ordered exponential
\begin{equation}
    U(t)=\mathcal{T}\exp \left(-i\int_0^t H(\{v_b(s)\})\D{s}\right)\,.
\end{equation}
We will consider a particular class of initial states $|\Psi(0)\rangle$ in \eqref{protocol}. We will assume that at time $t=0$ the initial state is a wavelet state
\begin{equation}
    |\Psi(0)\rangle=|\pmb{\cg{\alpha}}\rangle\,,
\end{equation}
with $\pmb{\cg{\alpha}}\subset K_\otimes$ fixed.

\section{Exact dynamics \label{exact}}

\subsection{Change of basis}
Given an observable $\mathcal{O}$, its expectation value at time $t$
\begin{equation}
    \langle\mathcal{O}\rangle_t=\langle\Psi(t)|\mathcal{O}|\Psi(t)\rangle
\end{equation}
can be expressed the following way in the wavelet basis
\begin{equation}\label{timeout}
    \langle \mathcal{O}\rangle_t=\sum_{\pmb{\cg{\beta}},\pmb{\cg{\gamma}}\subset K_\otimes} \langle \pmb{\cg{\alpha}}|U(-t)|\pmb{\cg{\beta}}\rangle \langle \pmb{\cg{\beta}}|\mathcal{O}|\pmb{\cg{\gamma}}\rangle \langle \pmb{\cg{\gamma}}|U(t)|\pmb{\cg{\alpha}}\rangle\,.
\end{equation}
The form factors of the time evolution operator $\langle \pmb{\cg{\alpha}}|U(t)|\pmb{\cg{\beta}}\rangle$ are non-trivial, as a tensor product of eigenstates of $H_0$ in each box is not an eigenstate of $H_0$ in the entire system (and even less so with inhomogeneous potentials). To compute these form factors $\langle \pmb{\cg{\alpha}}|U(t)|\pmb{\cg{\beta}}\rangle$, we first compute the overlap between eigenstates of $H_0$ and wavelet states.
\begin{lemma}[Overlaps]\label{overlap}
We have the following overlap between a wavelet state $|\pmb{\cg{\alpha}}\rangle$ and an eigenstate $|\pmb{\lambda}\rangle$ of $H_0$ with same number of particles $N$
\begin{equation}
\begin{aligned}
&\langle \pmb{\lambda}|\pmb{\cg{\alpha}}\rangle=\det_{i,j=1,...,N} \varphi(\lambda_j,\cg{\alpha}_k)\,,
\end{aligned}
\end{equation}
with
\begin{equation}
\begin{aligned}
\varphi(\lambda,\cg{\alpha})\equiv
\frac{i\sqrt{n}}{L} \cdot\frac{1-e^{iL\tfrac{\lambda}{n}}}{\lambda-\alpha}e^{-iL\tfrac{b(\cg{\alpha})}{n}\lambda}\,.
\end{aligned}
\end{equation}
If they have different number of particles, the form factor vanishes.
\end{lemma}
\begin{proof}
We have
\begin{equation}
\begin{aligned}
&\langle \pmb{\lambda}|\pmb{\cg{\alpha}}\rangle=\frac{n^{N/2}}{N!L^N\prod_{b=1}^n N_b!}\sum_{\pi\in\mathfrak{S}_{N}}(-1)^\pi \sum_{\sigma_1\in\mathfrak{S}_{N_1}} (-1)^{\sigma_1}...\sum_{\sigma_n\in\mathfrak{S}_{N_n}} (-1)^{\sigma_n}\\
&\int_0^{L/n} \D{x_1}...\D{x_{N_1}}...\int_{(n-1)L/n}^{L} \D{x_{N-N_n+1}}...\D{x_{N}}\int_{0}^{L} \D{x'_{1}}...\D{x'_{N}}\\
& \exp\left(-i\sum_{j=1}^{N} x'_j\lambda_{\pi j} \right)\prod_{i<j}\sign(x'_{i}-x'_{j}) \\
&\prod_{b=1}^n\exp\left(i\sum_{j=1}^{N_{b}} x_{N_1+...+N_{b-1}+j}\alpha^{(b)}_{\sigma_b j} \right)\prod_{i<j}\sign(x_{N_1+...+N_{b-1}+i}-x_{N_1+...+N_{b-1}+j}) \\
&\langle 0| \phi(x_1)...\phi(x_N)\phi^\dagger (x'_1)...\phi^\dagger (x'_N)|0\rangle\,.
\end{aligned}
\end{equation}
The last line imposes a matching $x'_j=x_{\kappa j}$ for $\kappa \in\mathfrak{S}_{N}$, and the two products of signs give a factor $(-1)^\kappa$. Then we write
\begin{equation}
\sum_{j=1}^{N} x'_j\lambda_{\pi j}=\sum_{j=1}^{N} x_j\lambda_{\pi \kappa^{-1}j}\,,
\end{equation}
and set $\pi=\pi' \kappa$ and then $\pi'\to \pi$ to obtain
\begin{equation}
\begin{aligned}
&\langle \pmb{\lambda}|\pmb{\cg{\alpha}}\rangle=\frac{n^{N/2}}{L^N\prod_{b=1}^n N_b!} \sum_{\sigma_1\in\mathfrak{S}_{N_1}} (-1)^{\sigma_1}...\sum_{\sigma_n\in\mathfrak{S}_{N_n}} (-1)^{\sigma_n}\sum_{\pi\in\mathfrak{S}_{N}}(-1)^\pi\\
&\int_0^{L/n} \D{x_1}...\D{x_{N_1}}...\int_{(n-1)L/n}^{L} \D{x_{N-N_n+1}}...\D{x_{N}}\\
& \exp\left(-i\sum_{j=1}^{N} x_j\lambda_{\pi j} \right) \prod_{b=1}^n\exp\left(i\sum_{j=1}^{N_{b}} x_{N_1+...+N_{b-1}+j}\alpha^{(b)}_{\sigma_b j} \right)\,.
\end{aligned}
\end{equation}
We now reparametrize $x_{N_1+...+N_{b-1}+j}$ into $ x_{N_1+...+N_{b-1}+\sigma_bj}$, and $\pi$ into $\pi \sigma_1...\sigma_n$. We obtain
\begin{equation}
\begin{aligned}
&\langle \pmb{\lambda}|\pmb{\cg{\alpha}}\rangle=\frac{n^{N/2}}{L^N}\sum_{\pi\in\mathfrak{S}_{N}}(-1)^\pi\int_0^{L/n} \D{x_1}...\D{x_{N_1}}...\int_{(n-1)L/n}^{L} \D{x_{N-N_n+1}}...\D{x_{N}}\\
& \exp\left(-i\sum_{j=1}^{N} x_j\lambda_{\pi j} \right) \prod_{b=1}^n\exp\left(i\sum_{j=1}^{N_{b}} x_{N_1+...+N_{b-1}+j}\alpha_{j}^{(b)} \right)\,.
\end{aligned}
\end{equation}
We now perform the integrals, and using $\alpha^{(b)}_j\in K$ we obtain
\begin{equation}
\begin{aligned}
&\langle \pmb{\lambda}|\pmb{\cg{\alpha}}\rangle=\frac{n^{N/2}}{L^N}\sum_{\pi\in\mathfrak{S}_{N}}(-1)^\pi
\prod_{j=1}^N \frac{e^{-iL\tfrac{b(\cg{\alpha}_j)}{n}\lambda_{\pi j}}-e^{-iL\tfrac{b(\cg{\alpha}_j)-1}{n}\lambda_{\pi j}}}{i(\alpha_j-\lambda_{\pi j})}\,.
\end{aligned}
\end{equation}
This can be expressed as a determinant
\begin{equation}
\begin{aligned}
&\langle \pmb{\lambda}|\pmb{\cg{\alpha}}\rangle=\det_{j,k=1,...,N} \varphi(\lambda_j,\cg{\alpha}_k)\,,
\end{aligned}
\end{equation}
with $\varphi$ given in the Lemma.
\end{proof}

\subsection{Form factors of the time evolution operator in the wavelet basis}
We now prove the following result.

\begin{theorem}[Form factors of $U(t)$]\label{eith}
We have the following form factor of $U(t)$ between two wavelet states
\begin{equation}
\begin{aligned}
&\langle \pmb{\cg{\alpha}}|U(t)|\pmb{\cg{\beta}}\rangle=\det_{i,j} \psi_t(\cg{\alpha}_i,\cg{\beta}_j)\,.
\end{aligned}
\end{equation}
Here, $\psi_t(\cg{\alpha},\cg{\beta})$ is the function of $\cg{\alpha},\cg{\beta}\in K_\otimes$ given by
\begin{equation}
    \psi_t(\cg{\alpha},\cg{\beta})=\sum_{\cg{\gamma}\in K_\otimes}G_{t}(\cg{\alpha},\cg{\gamma})\tilde{\psi}_t(\cg{\gamma},\cg{\beta})\,,
\end{equation}
with
\begin{equation}
    G_t(\cg{\alpha},\cg{\beta})=\sum_{\lambda\in K}\varphi(\lambda,\cg{\alpha})^*\varphi(\lambda,\cg{\beta})e^{-it\lambda^2}\,,
\end{equation}
and where $\tilde{\psi}_t(\cg{\alpha},\cg{\beta})$ satisfies the differential equation
\begin{equation}
    \partial_t\tilde{\psi}_t(\cg{\alpha},\cg{\beta})=-i\sum_{\cg{\nu}\in K_\otimes}U_t(\cg{\alpha},\cg{\nu})\tilde{\psi}_t(\cg{\nu},\cg{\beta})\,,
\end{equation}
with
\begin{equation}\label{utalphabeta}
    U_t(\cg{\alpha},\cg{\beta})=\sum_{\cg{\gamma}\in K_\otimes}G_{-t}(\cg{\alpha},\cg{\gamma})v_{b(\cg{\gamma})}(t)G_{t}(\cg{\gamma},\cg{\beta})\,,
\end{equation}
with initial condition $\tilde{\psi}_0(\cg{\alpha},\cg{\beta})=\delta_{\cg{\alpha},\cg{\beta}}$.
\end{theorem}
\begin{proof}
Using Trotter's expansion we have
\begin{equation}
    U(t)=\underset{T\to\infty}{\lim}\, \prod_{m=1}^T\left(e^{-i\frac{t}{T}H_0} e^{-i\frac{t}{T}V(\tfrac{mt}{T})}\right)\,.
\end{equation}
Let us fix $T$, introduce $\delta t=\tfrac{t}{T}$ and define
\begin{equation}
    F_m(\pmb{\cg{\alpha}},\pmb{\cg{\beta}})=\langle \pmb{\cg{\alpha}}| \prod_{m'=1}^m\left(e^{-i\delta tH_0} e^{-i\delta tV(m'\delta t)}\right)|\pmb{\cg{\beta}}\rangle\,.
\end{equation}
Let us show that we can always write
\begin{equation}\label{rec}
F_m(\pmb{\cg{\alpha}},\pmb{\cg{\beta}})=\det_{i,j}\psi_t^{(m)}(\cg{\alpha}_i,\cg{\beta}_j)\,,
\end{equation}
with some function $\psi_t^{(m)}(\cg{\alpha},\cg{\beta})$. For $m=0$, this is true with $\psi_t^{(0)}(\cg{\alpha},\cg{\beta})=\delta_{\cg{\alpha},\cg{\beta}}$. Let us assume it is true for $m-1$ and show it is true for $m$. We have
\begin{equation}
\begin{aligned}
     F_{m}(\pmb{\cg{\alpha}},\pmb{\cg{\beta}})&=\langle \pmb{\cg{\alpha}}| e^{-i\delta tH_0} e^{-i\delta tV(m\delta t)}\prod_{m'=1}^{m-1}\left(e^{-i\delta tH_0} e^{-i\delta tV(m'\delta t)}\right)|\pmb{\cg{\beta}}\rangle\\
    &=\sum_{\pmb{\lambda},\pmb{\cg{\gamma}}}\langle \pmb{\cg{\alpha}}| e^{-i\delta tH_0} |\pmb{\lambda}\rangle\langle\pmb{\lambda}|e^{-i\delta tV(m\delta t)}|\pmb{\cg{\gamma}}\rangle\langle\pmb{\cg{\gamma}}|\prod_{m'=1}^{m-1}\left(e^{-i\delta tH_0} e^{-i\delta tV(m'\delta t)}\right)|\pmb{\cg{\beta}}\rangle\\
    &=\sum_{\pmb{\lambda},\pmb{\cg{\gamma}}}\langle \pmb{\cg{\alpha}}| \pmb{\lambda}\rangle\langle\pmb{\lambda}|\pmb{\cg{\gamma}}\rangle\langle \pmb{\cg{\alpha}}| \prod_{m'=1}^{m-1}\left(e^{-i\delta tH_0} e^{-i\delta tV(m'\delta t)}\right)|\pmb{\cg{\beta}}\rangle e^{-i\delta t\sum_i \lambda_i^2}e^{-i\delta t\sum_{i}v_{b(\cg{\gamma}_i)}(m\delta t)}\,.
\end{aligned}
\end{equation}
We now use Lemma \ref{overlap} to express the form factors $\langle \pmb{\cg{\alpha}}| \pmb{\lambda}\rangle$ as determinants. We then use the following Lemma to sum over $\pmb{\lambda}$, proven e.g. in \cite{dreyer2021quantum,granet2022out}.
\begin{lemma}[Andreief identity]\label{andreief}
Given two functions $f(\lambda,\mu)$ and $g(\lambda,\mu)$, a set $S$ and two sets of numbers $\lambda_1,...,\lambda_N$ and $\mu_1,...,\mu_N$, we have the relation
\begin{equation}
    \sum_{k_1<...<k_N\in S}\det_{i,j}[f(\lambda_i,k_j)]\det_{i,j}[g(k_i,\mu_j)]=\det_{i,j}\left[\sum_{k\in K}f(\lambda_i,k)g(k,\mu_j) \right]\,.
\end{equation}
\end{lemma}
We obtain
\begin{equation}
    F_{m}(\pmb{\cg{\alpha}},\pmb{\cg{\beta}})=\sum_{\pmb{\cg{\gamma}}} \det_{i,j}\left[\sum_{\lambda\in K}\varphi(\cg{\alpha}_i,\lambda)^*\varphi(\lambda,\cg{\gamma}_j)e^{-i\delta t\lambda^2} \right] F_{m-1}(\pmb{\cg{\gamma}},\pmb{\cg{\beta}})e^{-i\delta t\sum_{i}v_{b(\cg{\gamma}_i)}(m\delta t)}\,.
\end{equation}
Then using the recurrence assumption on $F_m(\pmb{\cg{\gamma}},\pmb{\cg{\beta}})$ we use again Lemma \ref{andreief} to sum over $\pmb{\cg{\gamma}}$. This yields
\begin{equation}
     F_{m}(\pmb{\cg{\alpha}},\pmb{\cg{\beta}})=\det_{i,j}\psi_t^{(m)}(\cg{\alpha}_i,\cg{\beta}_j)\,,
\end{equation}
with
\begin{equation}\label{diff}
    \psi_t^{(m)}(\cg{\alpha},\cg{\beta})= \sum_{\cg{\gamma}\in K_\otimes} \sum_{\lambda\in K}\varphi(\cg{\alpha},\lambda)^*\varphi(\lambda,\cg{\gamma}) \psi_t^{(m-1)}(\cg{\gamma},\cg{\beta})e^{-i\delta t\lambda^2}e^{-i\delta t v_{b(\cg{\gamma})}(m\delta t)}\,.
\end{equation}
Hence by recurrence the form \eqref{rec} is true for all $m$.

To obtain a well-defined differential equation in the Trotter limit $\delta t\to 0$, we go to the interaction picture by introducing the Green's function for $\cg{\alpha},\cg{\beta}\in K_\otimes$
\begin{equation}\label{GTab}
    G_t(\cg{\alpha},\cg{\beta})=\sum_{\lambda\in K}\varphi(\lambda,\cg{\alpha})^*\varphi(\lambda,\cg{\beta})e^{-it\lambda^2}\,,
\end{equation}
and defining $\tilde{\psi}_t^{(m)}(\cg{\alpha},\cg{\beta})$ by
\begin{equation}
    \tilde{\psi}_t^{(m)}(\cg{\alpha},\cg{\beta})=\sum_{\cg{\gamma}\in K_\otimes}G_{-m\delta t}(\cg{\alpha},\cg{\gamma})\psi_t^{(m)}(\cg{\gamma},\cg{\beta})\,.
\end{equation}
We note that we have
\begin{equation}
    \sum_{\cg{\gamma}\in K_\otimes}G_t(\cg{\alpha},\cg{\gamma})G_{-t}(\cg{\gamma},\cg{\beta})=\delta_{\cg{\alpha},\cg{\beta}}\,.
\end{equation}
From \eqref{diff} we obtain the recurrence relation on $\tilde{\psi}_t^{(m)}$
\begin{equation}
    \tilde{\psi}_t^{(m)}(\cg{\alpha},\cg{\beta})=\sum_{\substack{\cg{\gamma},\cg{\kappa},\cg{\nu}\in K_\otimes\\ \lambda\in K}}G_{-m\delta t}(\cg{\alpha},\cg{\kappa})\varphi(\lambda,\cg{\kappa})^* e^{-i\delta t\lambda^2}\varphi(\lambda,\cg{\gamma})e^{-i\delta t v_{b(\cg{\gamma})}(m\delta t)}G_{(m-1)\delta t}(\cg{\gamma},\cg{\nu})\tilde{\psi}_t^{(m-1)}(\cg{\nu},\cg{\beta})\,.
\end{equation}
Performing the sum over $\cg{\kappa}$ using
\begin{equation}
    \sum_{\cg{\gamma}\in K_\otimes}\varphi(\lambda,\cg{\gamma})^*\varphi(\mu,\cg{\gamma})=\delta_{\lambda,\mu}\,,
\end{equation}
we obtain
\begin{equation}
    \tilde{\psi}_t^{(m)}(\cg{\alpha},\cg{\beta})=\sum_{\cg{\nu}\in K_\otimes}\left(\sum_{\cg{\gamma}\in K_\otimes}G_{-(m-1)\delta t}(\cg{\alpha},\cg{\gamma})e^{-i\delta t v_{b(\cg{\gamma})}(m\delta t)}G_{(m-1)\delta t}(\cg{\gamma},\cg{\nu})\right)\tilde{\psi}_t^{(m-1)}(\cg{\nu},\cg{\beta})\,.
\end{equation}
This form allows for a well-defined Trotter limit $\delta t\to 0$. We thus obtain the representation
\begin{equation}
\begin{aligned}
&\langle \pmb{\cg{\alpha}}|U(t)|\pmb{\cg{\beta}}\rangle=\det_{i,j} \psi_t(\cg{\alpha}_i,\cg{\beta}_j)\,,
\end{aligned}
\end{equation}
with $\psi_t(\cg{\alpha},\cg{\beta})$ given by
\begin{equation}
    \psi_t(\cg{\alpha},\cg{\beta})=\sum_{\cg{\gamma}\in K_\otimes}G_{t}(\cg{\alpha},\cg{\gamma})\tilde{\psi}_t(\cg{\gamma},\cg{\beta})\,,
\end{equation}
with $\tilde{\psi}_t(\cg{\alpha},\cg{\beta})$ satisfying the differential equation
\begin{equation}
    \partial_t\tilde{\psi}_t(\cg{\alpha},\cg{\beta})=-i\sum_{\cg{\nu}\in K_\otimes}\left(\sum_{\cg{\gamma}\in K_\otimes}G_{-t}(\cg{\alpha},\cg{\gamma})v_{b(\cg{\gamma})}(t)G_{t}(\cg{\gamma},\cg{\nu})\right)\tilde{\psi}_t(\cg{\nu},\cg{\beta})\,.
\end{equation}
This concludes our proof.

\end{proof}

The reader could notice that in Theorem \ref{eith} we could have written a differential equation directly in terms of $\psi_t$ rather than $\tilde{\psi}_t$. In fact, the matrix appearing in the differential equation for $\psi_t$ would be defined in terms of a divergent series. This reflects the fact that $\psi_t$ itself is not differentiable at $t=0$. This can also be seen at the level of $G_t$, whose expression in terms of a series \eqref{GTab} does not allow for term-by-term time differentiation.

\section{Expectation values of local operators \label{expe}}
\subsection{Cumulative reduced density matrix}
We now would like to express expectation values of observables localized in box $b$ within state $|\Psi(t)\rangle$ in terms of $\psi_t(\cg{\alpha},\cg{\beta})$. To that end, we need to compute the reduced density matrix in box $b$
\begin{equation}
    \rho_b=\underset{\substack{c=1,...,n\\ c\neq b}}{\Tr}\big(|\Psi(t)\rangle\langle\Psi(t)|\big)
\end{equation}
where the trace is over all the boxes $c\neq b$. As we will see, another related quantity has a simpler expression in terms of $\psi_t(\cg{\alpha},\cg{\beta})$. This is the cumulative reduced density matrix that we define by
\begin{equation}\label{rhocumuldef}
    \rho_b^{\rm cumul}=\sum_{\substack{\pmb{\cg{x}},\pmb{\cg{y}}\subset K_\otimes\\ b(\pmb{\cg{x}})=b(\pmb{\cg{y}})=b}}|\pmb{\cg{x}}\rangle\langle \pmb{\cg{y}}| \sum_{\substack{\pmb{\cg{\beta}},\pmb{\cg{\gamma}}\subset K_\otimes \\ \pmb{\cg{x}}\subset \pmb{\cg{\beta}}\,,\,\pmb{\cg{y}}\subset \pmb{\cg{\gamma}} \\ \pmb{\cg{\beta}}\setminus \pmb{\cg{x}}=\pmb{\cg{\gamma}}\setminus \pmb{\cg{y}}}} \langle \pmb{\cg{\beta}}|\Psi(t)\rangle \langle \Psi(t)|\pmb{\cg{\gamma}}\rangle\,,
\end{equation}
where by $b(\pmb{\cg{x}})=b$ we mean that for all $\cg{x}\in\pmb{\cg{x}}$ we have $b(\cg{x})=b$. We note that from this expression, the definition of the reduced density matrix $\rho_b$ would be obtained by imposing that $\pmb{\cg{\beta}}\setminus\pmb{\cg{x}}$ and $\pmb{\cg{\gamma}}\setminus\pmb{\cg{y}}$ have no particles in box $b$. The terms for which $\pmb{\cg{\beta}}\setminus\pmb{\cg{x}}$ and $\pmb{\cg{\gamma}}\setminus\pmb{\cg{y}}$ have exactly one particle in box $b$ can be expressed in terms of the reduced density matrix as $T[\rho_b]$ with
the linear operator applying on density matrices
\begin{equation}
    T[\rho]=\sum_{\lambda\in K_{\rm box}}\phi(\lambda)\rho\phi^\dagger(\lambda)\,,
\end{equation}
with $\phi^\dagger(\lambda)$ the bosonic creation operator for the mode $\lambda\in K_{\rm box}$ in box $b$. Similarly, terms with higher numbers of particles can be expressed through powers of $T$. One thus has
\begin{equation}\label{rhocumulformula}
     \rho_b^{\rm cumul}=\sum_{k=0}^\infty \frac{T^k[\rho_b]}{k!}\,.
\end{equation}
The inverse of this relation is then seen to be
\begin{equation}\label{rhotorhocumul}
    \rho_b=\sum_{k=0}^\infty (-1)^k\frac{T^k[\rho_b^{\rm cumul}]}{k!}\,.
\end{equation}
Hence knowing $ \rho_b^{\rm cumul}$ or $\rho_b$ is equivalent.

\subsection{Expectation values within a box in terms of \texorpdfstring{$\rho_b^{\rm cumul}$}{Lg}}
In fact, the expectation value of many local operators can be directly expressed in terms of $\rho_b^{\rm cumul}$, without using $\rho_b$. For example, let us take the example of the two-point function of the boson density operator
\begin{equation}\label{2pt}
\mathcal{O}_b=\phi^\dagger(x)\phi(x)\phi^\dagger(y)\phi(y)\,,
\end{equation}
with $x,y\in [(b-1)\tfrac{L}{n},b\tfrac{L}{n}]$ localized in box $b$. Its expectation value can be expressed through $\rho_b^{\rm cumul}$ as
\begin{equation}
    \langle \mathcal{O}_b \rangle=\frac{1}{L^2}\sum_{h_1,h_2,p_1,p_2\in K_{\rm box}} e^{ix(p_1-h_1)+iy(p_2-h_2)}{}_b \langle h_1,h_2|\rho_b^{\rm cumul}|p_1,p_2\rangle_b\,,
\end{equation}
with $|p_1,p_2\rangle_b$ denoting the state with only two particles $p_1,p_2$ in box $b$.\\

Let us define observables that allow us to probe the particle content of a box $b$. Given a function $w(\beta)$, we define the operator $\mathcal{O}_b(w)$ as being localized in box $b$ and being diagonal in the wavelet basis $|\pmb{\cg{\alpha}}\rangle$ with eigenvalue
\begin{equation}\label{obw}
    \mathcal{O}_b(w)|\pmb{\cg{\alpha}}\rangle= \sum_{\substack{\cg{\alpha}\in\pmb{\cg{\alpha}}\\ b(\cg{\alpha})=b}}w(\cg{\alpha})|\pmb{\cg{\alpha}}\rangle\,.
\end{equation}
The expectation value of $\mathcal{O}_b(w)$ in state $|\Psi(t)\rangle$ is readily expressed in terms of $\rho_b^{\rm cumul}$
\begin{equation}\label{obwcumul}
    \langle\mathcal{O}_b(w)\rangle_t=\sum_{\alpha\in K_{\rm box}}w(\alpha){}_b\langle \alpha|\rho_b^{\rm cumul}|\alpha\rangle_b\,,
\end{equation}
where $|\alpha\rangle_b$ denotes a single particle eigenstate with momentum $\alpha$ in box $b$.

\subsection{Expression of \texorpdfstring{$\rho_b^{\rm cumul}$}{Lg} in terms of \texorpdfstring{$\psi_t(\cg{\alpha},\cg{\beta})$}{Lg}}
This cumulative reduced density matrix admits the following simple expression in terms of $\psi_t(\cg{\alpha},\cg{\beta})$.
\begin{theorem}\label{theocumul}
    For $\pmb{\cg{x}},\pmb{\cg{y}}\subset K_\otimes$ localized in box $b$ with $M$ particles we have
    \begin{equation}
        \langle \pmb{\cg{x}}|\rho_b^{\rm cumul}|\pmb{\cg{y}}\rangle=\sum_{q_1,...,q_M=1}^N\det_{i,j=1,...,M}\left[ \psi_t(\cg{\alpha}_{q_i},\cg{x}_i)\psi_t(\cg{\alpha}_{q_j},\cg{y}_i)^* \right]\,,
    \end{equation}
    where we recall that $|\pmb{\cg{\alpha}}\rangle$ denotes the initial wavelet state. If $\pmb{\cg{x}},\pmb{\cg{y}}$ have a different number of particles, then $\langle \pmb{\cg{x}}|\rho_b^{\rm cumul}|\pmb{\cg{y}}\rangle=0$.
\end{theorem}
\begin{proof}
Firstly, we see in the definition of $\rho_b^{\rm cumul}$ in \eqref{rhocumuldef} that $\pmb{\cg{\beta}}$ and $\pmb{\cg{\gamma}}$ must have the same number of particles $N$ as in $|\Psi(t)\rangle$. Since $\pmb{\cg{\beta}}\setminus \pmb{\cg{x}}=\pmb{\cg{\gamma}}\setminus \pmb{\cg{y}}$, there has to be the same number of particles in $\pmb{\cg{x}}$ and $\pmb{\cg{y}}$.

From the definition we have
\begin{equation}
     \langle \pmb{\cg{x}}|\rho_b^{\rm cumul}|\pmb{\cg{y}}\rangle=\sum_{\substack{\pmb{\cg{\beta}} \subset K_\otimes\\ \pmb{\cg{\beta}}\cap \pmb{\cg{x}}=\pmb{\cg{\beta}}\cap \pmb{\cg{y}}=\emptyset}} \langle \pmb{\cg{\beta}}\cup \pmb{\cg{x}}|U(t)|\pmb{\cg{\alpha}}\rangle\langle\pmb{\cg{\alpha}}|U(-t)|\pmb{\cg{\beta}}\cup \pmb{\cg{y}}\rangle\,.
\end{equation}
Using Lemma \ref{eith} to express the form factors of $U(t)$ as determinants we have
\begin{equation}
    \langle \pmb{\cg{x}}|\rho_b^{\rm cumul}|\pmb{\cg{y}}\rangle=\sum_{\substack{\pmb{\cg{\beta}} \subset K_\otimes\\|\pmb{\cg{\beta}}|=N-M\\ \pmb{\cg{\beta}}\cap \pmb{\cg{x}}=\pmb{\cg{\beta}}\cap \pmb{\cg{y}}=\emptyset}} \det_{i,j} \psi_t(\cg{\beta}'_i,\cg{\alpha}_j) \det_{i,j} \psi_t(\cg{\beta}''_i,\cg{\alpha}_j)^*\,,
\end{equation}
with $\pmb{\cg{\beta}'}=\pmb{\cg{\beta}}\cup\pmb{\cg{x}}$ and $\pmb{\cg{\beta}''}=\pmb{\cg{\beta}}\cup\pmb{\cg{y}}$ imposed to have $N$ particles. In this expression we can include in the sum the terms where $\pmb{\cg{\beta}}$ contains elements of $\pmb{\cg{x}}$ or $\pmb{\cg{y}}$, since the determinants vanish in this case. Hence, setting $\pmb{\cg{x}}$ and $\pmb{\cg{y}}$ to be the first $M$ indices of $\pmb{\cg{\beta}'},\pmb{\cg{\beta}''}$ and writing $\pmb{\cg{\beta}}=\{\cg{\beta}_{M+1},...,\cg{\beta}_N\}$ we have
\begin{equation}
     \langle \pmb{\cg{x}}|\rho_b^{\rm cumul}|\pmb{\cg{y}}\rangle=\frac{1}{(N-M)!}\sum_{\cg{\beta}_{M+1}\in K_\otimes}...\sum_{\cg{\beta}_{N}\in K_\otimes}\det_{i,j} \psi_t(\cg{\beta}'_i,\cg{\alpha}_j) \det_{i,j} \psi_t(\cg{\beta}''_i,\cg{\alpha}_j)^*\,.
\end{equation}
Then we write that the product of determinants is the determinant of the product
\begin{equation}
   \det_{i,j} \psi_t(\cg{\beta}'_i,\cg{\alpha}_j) \det_{i,j} \psi_t(\cg{\beta}''_i,\cg{\alpha}_j)^*=\det_{i,j}\left[ \sum_{q=1}^N \psi_t(\cg{\beta}'_q,\cg{\alpha}_i)\psi_t(\cg{\beta}''_q,\cg{\alpha}_j)^*\right]\,,
\end{equation}
and expand the determinant as a sum over permutations
\begin{equation}
    \det_{i,j}\left[ \sum_{q=1}^N \psi_t(\cg{\beta}'_q,\cg{\alpha}_i)\psi_t(\cg{\beta}''_q,\cg{\alpha}_j)^*\right]=\sum_{\sigma\in\mathfrak{S}_N}(-1)^\sigma \prod_{i=1}^N \left(  \sum_{q=1}^N \psi_t(\cg{\beta}'_q,\cg{\alpha}_i)\psi_t(\cg{\beta}''_q,\cg{\alpha}_{\sigma(i)})^* \right)\,.
\end{equation}
Expanding the product, we see that if we pick twice the same $q$ for $i$ and $j$, then the product of the two terms is invariant under swapping $\sigma(i)$ and $\sigma(j)$, whereas this changes the sign of $(-1)^\sigma$, making the contribution vanish. Hence we have
\begin{equation}
     \det_{i,j}\left[ \sum_{q=1}^N \psi_t(\cg{\beta}'_q,\cg{\alpha}_i)\psi_t(\cg{\beta}''_q,\cg{\alpha}_j)^*\right]=\sum_{\sigma,\tau\in\mathfrak{S}_N}(-1)^\sigma \prod_{i=1}^N \psi_t(\cg{\beta}'_{\tau(i)},\cg{\alpha}_i)\psi_t(\cg{\beta}''_{\tau(i)},\cg{\alpha}_{\sigma(i)})^*\,.
\end{equation}
So
\begin{equation}
\begin{aligned}
     \langle \pmb{\cg{x}}|\rho_b^{\rm cumul}|\pmb{\cg{y}}\rangle=\frac{1}{(N-M)!}\sum_{\sigma,\tau\in\mathfrak{S}_N}&(-1)^\sigma \prod_{i=1}^M \psi_t(\cg{x}_i,\cg{\alpha}_{\tau^{-1}(i)})\psi_t(\cg{y}_i,\cg{\alpha}_{\sigma\tau^{-1}(i)})^*\\
     &\times\sum_{\cg{\beta}_{M+1}\in K_\otimes}...\sum_{\cg{\beta}_{N}\in K_\otimes} \prod_{i=M+1}^N \psi_t(\cg{\beta}_i,\cg{\alpha}_{\tau^{-1}(i)})\psi_t(\cg{\beta}_i,\cg{\alpha}_{\sigma\tau^{-1}(i)})^*\,.
\end{aligned}
\end{equation}
The sum over $\cg{\beta}_{M+1},...,\cg{\beta}_N$ factorizes into products of single sums that read
\begin{equation}\label{equal}
    \sum_{\cg{\beta}\in K_\otimes}\psi_t(\cg{\beta},\cg{\alpha}_{\tau^{-1}(i)})\psi_t(\cg{\beta},\cg{\alpha}_{\sigma\tau^{-1}(i)})^*=\delta_{\tau^{-1}(i),\sigma \tau^{-1}(i)}\,.
\end{equation}
Then the sum over $\tau$ does not depend on the ordering of $\tau^{-1}(M+1),...,\tau^{-1}(N)$. It can be converted into a sum over subsets $\{q_1,...,q_M\}\subset\{1,...,N\}$ with a factor $(N-M)!$ and with a remaining permutation over $\tau^{-1}(1),...,\tau^{-1}(M)$. It yields
\begin{equation}
     \langle \pmb{\cg{x}}|\rho_b^{\rm cumul}|\pmb{\cg{y}}\rangle=\sum_{\substack{\{q_1,...,q_M\}\\ \subset \{1,...,N\}}}\sum_{\sigma,\tau\in\mathfrak{S}_M}(-1)^\sigma \prod_{i=1}^M \psi_t(\cg{x}_i,\cg{\alpha}_{q_{\tau(i)}})\psi_t(\cg{y}_i,\cg{\alpha}_{q_{\sigma\tau(i)}})^*\,.
\end{equation}
We used that $\sigma$ leaves invariant the set $\{q_1,...,q_M\}$ because of \eqref{equal}, to write $q_{\sigma(i)}$ instead of $\sigma(q_i)$ after a reparametrisation of $\sigma$. The sum over $\sigma$ is exactly a $M\times M$ determinant. Then it vanishes if two $q_i$'s are equal, so the sum over the subset $\{q_1,...,q_M\}$ can be converted into $M$ sums over $q\in K$ with a factor $1/M!$, and the sum over $\tau$ gives a factor $M!$. Hence
\begin{equation}
     \langle \pmb{\cg{x}}|\rho_b^{\rm cumul}|\pmb{\cg{y}}\rangle=\sum_{q_1,...,q_M\in \{1,...,N\}}\det_{i,j} [\psi_t(\cg{x}_i,\cg{\alpha}_{q_i})\psi_t(\cg{y}_i,\cg{\alpha}_{q_j})^*]\,,
\end{equation}
which concludes the proof.
\end{proof}

\subsection{Expectation values between different boxes}
\subsubsection{Generalities}
The previous construction is straightforwardly generalized to observables in several boxes. Given two boxes $b_1,b_2$, we define the reduced density matrix
\begin{equation}
    \rho_{b_1b_2}=\underset{\substack{c=1,...,n\\ c\neq b_1,b_2}}{\Tr}\big(|\Psi(t)\rangle\langle\Psi(t)|\big)\,,
\end{equation}
and the cumulative reduced density matrix
\begin{equation}
    \rho_{b_1b_2}^{\rm cumul}=\sum_{\substack{\pmb{\cg{x}},\pmb{\cg{y}}\subset K_\otimes\\ b(\pmb{\cg{x}}),b(\pmb{\cg{y}})\in\{b_1,b_2\}}}|\pmb{\cg{x}}\rangle\langle \pmb{\cg{y}}| \sum_{\substack{\pmb{\cg{\beta}},\pmb{\cg{\gamma}}\subset K_\otimes \\ \pmb{\cg{x}}\subset \pmb{\cg{\beta}}\,,\,\pmb{\cg{y}}\subset \pmb{\cg{\gamma}} \\ \pmb{\cg{\beta}}\setminus \pmb{\cg{x}}=\pmb{\cg{\gamma}}\setminus \pmb{\cg{y}}}} \langle \pmb{\cg{\beta}}|\Psi(t)\rangle \langle \Psi(t)|\pmb{\cg{\gamma}}\rangle\,,
\end{equation}
where $b(\pmb{\cg{x}})\in\{b_1,b_2\}$ means that for all $\cg{x}\in\pmb{\cg{x}}$, we have $b(\cg{x})\in\{b_1,b_2\}$. A formula identical to \eqref{rhocumulformula} relates $\rho_{b_1b_2}^{\rm cumul}$ to $\rho_{b_1b_2}$. Expectation values of operators localized in boxes $b_1,b_2$ can also be written in terms of $\rho_{b_1b_2}^{\rm cumul}$. For example, with $\mathcal{O}_b(w)$ denoting the operator in \eqref{obw}, we have
\begin{equation}
    \langle \mathcal{O}_{b_1}(w_1)\mathcal{O}_{b_2}(w_2)\rangle=\sum_{\alpha_1,\alpha_2\in K_{\rm box}} w_1(\alpha_1)w_2(\alpha_2){}_{b_1}\langle \alpha_1|\otimes {}_{b_2}\langle\alpha_2|\rho_{b_1b_2}^{\rm cumul}|\alpha_1\rangle_{b_1}\otimes|\alpha_2\rangle_{b_2}\,.
\end{equation}
Then, an identical formula to Theorem \ref{theocumul} gives the expression for $\langle \pmb{\bar{x}}| \rho_{b_1b_2}|\pmb{\bar{y}}\rangle$ in terms of $\psi_t(\bar{\alpha},\bar{\beta})$.

\subsubsection{Single-particle Green's function }
The wavelet representation allows for an efficient computation of certain observables which would be difficult to compute otherwise. Let us for example consider the single-particle bosonic Green's function
\begin{equation}
    \mathcal{G}_t(x,y)=\langle \phi^\dagger(x) \phi(y)\rangle_t\,.
\end{equation}
Within an eigenstate of the entire system $\pmb{\lambda}$, this Green's function does not have a simple expression, because of the products of signs in \eqref{chi} \cite{wilson2020observation,rigol2005fermionization,papenbrock2003ground,minguzzi2002high,girardeau2001ground,pezer2007momentum}. Similarly, if $x,y$ are two coordinates within the same box, no simple expression exists. However, if the two positions correspond to the beginning of two different boxes, i.e.\ if we set $x=(b_1-1)n/L$ and $y=(b_2-1)n/L$, then we have from \eqref{chiwavelet} for two wavelet states $\pmb{\cg{\alpha}},\pmb{\cg{\beta}}$
\begin{equation}
    \begin{aligned}
        \langle \pmb{\cg{\alpha}}|\phi^\dagger(x)\phi(y)|\pmb{\cg{\beta}}\rangle=\frac{1}{L}\,,
    \end{aligned}
\end{equation}
if for all $b\neq b_1,b_2$, we have $\pmb{\alpha^{(b)}}=\pmb{\beta^{(b)}}$, and if $\pmb{\alpha^{(b_1)}}=\pmb{\beta^{(b_1)}}\cup \{\lambda\}$ and $\pmb{\beta^{(b_2)}}=\pmb{\alpha^{(b_2)}}\cup \{\mu\}$ for some $\lambda,\mu\in K_{\rm box}$. If this is not satisfied, then $ \langle \pmb{\cg{\alpha}}|\phi^\dagger(x)\phi(y)|\pmb{\cg{\beta}}\rangle=0$. We recall that $\pmb{\alpha^{(b)}}$ means the set of all particles in $\pmb{\cg{\alpha}}$ that belong to box $b$.

This yields the following expression in terms of the cumulative reduced density matrix in boxes $b_1,b_2$
\begin{equation}
     \mathcal{G}_t(x,y)=\frac{1}{L}\sum_{\alpha,\beta\in K_{\rm box}}{}_{b_1}\langle \alpha| \otimes {}_{b_2}\langle 0|\rho^{\rm cumul}_{b_1,b_2}|0\rangle_{b_1} \otimes |\beta\rangle_{b_2}\,.
\end{equation}
In terms of $\psi_t(\bar{\alpha},\bar{\beta})$, this is thus
\begin{equation}\label{singleparticle}
     \mathcal{G}_t(x,y)=\frac{1}{L}\sum_{q=1}^M \sum_{\substack{\cg{\lambda}\in K_\otimes\\ b(\cg{\lambda})=b_1}}\sum_{\substack{\cg{\mu}\in K_\otimes\\ b(\cg{\mu})=b_2}} \psi_t(\cg{\alpha}_q,\cg{\lambda})\psi_t(\cg{\alpha}_q,\cg{\mu})^*\,,
\end{equation}
where we recall that the $\cg{\alpha}_q$'s are the particles in the initial state.

\section{Hydrodynamic limit \label{hydrosec}}

\subsection{Definition}
In order to define the hydrodynamic limit that we study in this paper, let us introduce typical scales for the parameters of the problem. We define $\tau$ the typical time scale in the problem and $T$ the rescaled time by
\begin{equation}
    t=\tau T\,.
\end{equation}
We define $2\pi\kappa$ the momentum spacing in each of the boxes, namely
\begin{equation}
    \kappa=\frac{n}{L}\,.
\end{equation}
 We define $\Lambda$ the typical momentum of the particles in the wavelet initial state $|\pmb{\cg{\alpha}}\rangle$. We note that we necessarily have $\Lambda\gtrsim\kappa$. We finally define $X$ the rescaled position by
 \begin{equation}
     X=\frac{x}{\tau\Lambda}\,,
 \end{equation}
 with $x\in[0,L]$ the original position, related to the box index $b$ by $b=\lfloor x\kappa\rfloor$.\\

 The typical scales $\tau$ and $\Lambda$ are parameters of the problem that we can choose freely. They correspond to the time scale at which we study the problem and to the particle content of the initial state. The status of $\kappa$ is more subtle. Since it directly depends on the number $n$ of boxes that we divide the system into, it is not per se a physical parameter of the system, but rather parametrizes the precision at which we want to describe the system. However, through the fact that we impose the initial condition to be a wavelet state, it also constrains the minimal typical distance at which the initial state varies.\\ 
 
 We define the \emph{hydrodynamic limit} by the two conditions
 \begin{equation}\label{hydro}
     (i)\quad\tau\kappa^2\to 0\,,\qquad\qquad (ii)\quad\tau\Lambda^2\to\infty\,.
 \end{equation}
The first condition means that we study the system at a time scale at which the energy levels within each box appear continuous. The second condition means that we study the system at ``large" times, the meaning of ``large" being set by the initial state.

These two conditions do not impose any behaviour of $\tau\kappa\Lambda$. As we will see, this quantity is the typical number of boxes that a particle moves through during a time interval of order $\tau$. If $\tau\kappa\Lambda\to 0$, the system will have no time dependence for the rescaled time $T$. If $\tau\kappa\Lambda=\mathcal{O}(1)$, the system is not divided into small enough boxes and will appear discontinuous at scales $X,T$. We will thus assume
 \begin{equation}\label{smooth}
     \tau\kappa\Lambda\to\infty\,,
 \end{equation}
 which will ensure that the system is smooth in terms of the rescaled space and time variable $X,T$.\\

The condition \eqref{hydro} can be met in quite different physical contexts. Let us mention in particular two situations. The first context (i) is that of the \emph{Euler scale}, corresponding to large $x,t$ at fixed $x/t$, which is the one mostly mentioned in the GHD literature. In our notations this means $\tau\to\infty$ and $\Lambda=\mathcal{O}(1)$. In this case the size of the boxes $1/\kappa$ becomes infinitely large. The second context (ii) is a \emph{short-time, high density limit}. Because of the Pauli principle, an initial state with high density of particles will also have a large typical particle momentum $\Lambda$. This context translates thus into $\tau\to 0$ and $\Lambda\to\infty$ at e.g.\ fixed $\tau\Lambda$. The smoothness condition \eqref{smooth} imposes then $\kappa\to\infty$, which is infinitely small boxes. No rescaling of space is needed in this context. It is similar  to a semi-classical expansion $\hbar\to 0$.

\subsection{Generalized Hydrodynamics}
Before taking the hydrodynamic limit of our model, let us very briefly summarize GHD, the hydrodynamic theory describing the Lieb-Liniger model \cite{castro2016emergent,bertini2016transport}. As any hydrodynamic theory, it postulates that as long as local operators are concerned, the state of the system at rescaled coordinates $X,T$ can be considered in an equilibrium state. In the hardcore boson gas, equilibrium states $|\pmb{\lambda}\rangle$ are characterized by a root density $\rho(\lambda)$, defined by the fact that there are $L \rho(\lambda)\D{\lambda}$ particles with momentum between $\lambda$ and $\lambda+\D{\lambda}$ in the thermodynamic limit. GHD provides a partial differential equation for the root density $\rho_{X,T}(\lambda)$ describing the equilibrium state at rescaled coordinates $X,T$. In case of hardcore bosons and in absence of external potentials, it reads
    \begin{equation}
        \partial_{T}\rho_{X,T}(\lambda)+2\lambda \partial_{X}\rho_{X,T}(\lambda)=0\,.
    \end{equation}
The solution can be readily expressed in terms of $\rho_{X,0}(\lambda)$ the initial root densities
\begin{equation}\label{ghdeq}
    \rho_{X,T}(\lambda)=\rho_{X-2\lambda T,0}(\lambda)\,.
\end{equation}
In presence of an external potential $V(X)$, the GHD equations are \cite{doyon2017note,bastianello2019generalized}
\begin{equation}
     \partial_{T}\rho_{X,T}(\lambda)+2\lambda \partial_{X}\rho_{X,T}(\lambda)=\partial_{X} V(X) \partial_\lambda\rho_{X,T}(\lambda)\,.
\end{equation}

\subsection{Hydrodynamic limit of the Green's function}

In order to determine the hydrodynamic limit of the expectation values determined in Section \ref{exact}, let us first determine the hydrodynamic limit of the Green's function $G_t(\cg{\alpha},\cg{\beta})$.

\begin{lemma}\label{hydrodynamicG}
    For $\alpha,\beta\in K_{\rm box}$ of order $\Lambda$, we have the following hydrodynamic limit
    \begin{equation}
        G_t(\cg{\alpha},\cg{\beta})=G_t^{\rm hydro}(\cg{\alpha},\cg{\beta})+\mathcal{O}\left(\frac{\kappa^2 \sqrt{\tau}}{\max(|\alpha-\beta|,\kappa)} \right)\,,
    \end{equation}
    with if $\alpha\neq\beta$
    \begin{equation}
        G_t^{\rm hydro}(\cg{\alpha},\cg{\beta})=
         \frac{\kappa}{i(\alpha-\beta)}\Big[e^{-it\alpha^2} \partial_{b(\cg{\alpha})-b(\cg{\beta}), \lfloor 2t\kappa\alpha\rfloor } -e^{-it\beta^2} \partial_{b(\cg{\alpha})-b(\cg{\beta}), \lfloor 2t\kappa\beta\rfloor} \Big]\,,
    \end{equation}
    and if $\alpha=\beta$
    \begin{equation}\label{gequal}
G_t^{\rm hydro}(\cg{\alpha},\cg{\beta})=e^{-it\alpha^2}(\tilde{\delta}_{b(\cg{\alpha})-b(\cg{\beta})-\lfloor 2t\kappa\alpha+1/2\rfloor}-(2t\kappa\alpha-\lfloor 2t\kappa\alpha+1/2\rfloor) \partial_{b(\cg{\alpha})-b(\cg{\beta}), \lfloor 2t\kappa\alpha\rfloor})\,.
\end{equation}
We used the notation $\lfloor x\rfloor$ for the integer part of $x$, and the discrete derivative defined by
\begin{equation}\label{derivdetal}
    \partial_{i,j}=\tilde{\delta}_{i-j}-\tilde{\delta}_{i-j-1}\,,
\end{equation}
with the $\delta$ function modulo $n$
\begin{equation}
    \tilde{\delta}_j=\delta_{j\,{\rm mod}\,n}(-1)^{j/n}\,.
\end{equation}

\end{lemma}
\begin{proof}
  We write
    \begin{equation}
        G_t(\cg{\alpha},\cg{\beta})=\sum_{q=0}^{n-1} \sum_{\lambda\in K_q}\varphi(\lambda,\cg{\alpha})^*\varphi(\lambda,\cg{\beta})e^{-it\lambda^2}\,,
    \end{equation}
with
\begin{equation}\label{kbarq}
 K_q=\left\{\frac{2\pi (nm+q+\tfrac{1}{2})}{L},\quad m\in\mathbb{Z}\right\}\,.
\end{equation}
We have
\begin{equation}
    \sum_{\lambda\in K_q}\varphi(\lambda,\cg{\alpha})^*\varphi(\lambda,\cg{\beta})e^{-it\lambda^2}=\frac{n}{L^2}e^{2i\pi (q+1/2)\frac{b(\cg{\alpha})-b(\cg{\beta})}{n}}4\sin^2(\tfrac{\pi (q+1/2)}{n}) \sum_{\lambda\in K_q}\frac{e^{-it\lambda^2}}{(\alpha-\lambda)(\beta-\lambda)}\,.
\end{equation}
Let us first treat the case $\alpha\neq\beta$. Then
\begin{equation}
    \sum_{\lambda\in K_q}\frac{e^{-it\lambda^2}}{(\alpha-\lambda)(\beta-\lambda)}=\frac{1}{\alpha-\beta}\sum_{\lambda\in K_q}\frac{e^{-it(\lambda+\alpha)^2}-e^{-it(\lambda+\beta)^2}}{\lambda}\,.
\end{equation}
Using
\begin{equation}\label{sum1}
    \sum_{m\in\mathbb{Z}}\frac{e^{ix(m+\theta)}}{m+\theta}=\frac{\pi }{\sin\pi\theta}e^{2i\pi\theta\left( \intg{\frac{x}{2\pi}}+\frac{1}{2}\right)}\,,
\end{equation}
with $\intg{x}$ the integer part of $x$, we decompose it as
\begin{equation}\label{decompo}
    \begin{aligned}
        & \sum_{\lambda\in K_q}\frac{e^{-it\lambda^2}}{(\alpha-\lambda)(\beta-\lambda)}\\
         &=\frac{Le^{-\frac{i\pi (q+1/2)}{n}}}{2n(\alpha-\beta)\sin(\tfrac{\pi (q+1/2)}{n})}\Big( e^{-it\alpha^2-\frac{2i\pi (q+1/2)}{n}\lfloor 2t\kappa\alpha\rfloor}-e^{-it\beta^2-\frac{2i\pi (q+1/2)}{n}\lfloor 2t\kappa\beta\rfloor}\Big)\\
         &+\frac{R_t(\alpha)-R_t(\beta)}{\alpha-\beta}\,,
    \end{aligned}
\end{equation}
with
\begin{equation}
    R_t(\alpha)= e^{-it\alpha^2}\sum_{\lambda\in K_q} \frac{(e^{-it\lambda^2}-1)e^{-2it\lambda\alpha}}{\lambda}\,.
\end{equation}
In the sum, in the hydrodynamic limit the values $\sqrt{t}\lambda$ are spaced by an order $\sqrt{\tau}\kappa$ which goes to $0$ from condition (i) of \eqref{hydro}. Hence the sum can be seen as $1/\kappa$ times a Riemann sum for $\sqrt{t}\lambda$. Since $u\mapsto\frac{e^{-iu^2}-1}{u}e^{-2i\sqrt{t}u\alpha}$ is square integrable but not its derivative, the Riemann sum converges as $\mathcal{O}(\kappa\sqrt{\tau})$. This yields
\begin{equation}
    R_t(\alpha)=e^{-it\alpha^2}\frac{L}{2\pi n}\int_{-\infty}^\infty \frac{e^{-iu^2}-1}{u}e^{-2i\sqrt{t}u\alpha}\D{u}+\mathcal{O}(\sqrt{\tau})\,.
\end{equation}
Moreover, since $u\mapsto\frac{e^{-iu^2}-1}{u}$ is square integrable but not its derivative, its Fourier transform as a function of $\omega$ decays as $1/\omega$. Hence because $\sqrt{t}$ scales as $\sqrt{\tau}$ and $\alpha$ as $\Lambda$ in the hydrodynamic limit, we have
\begin{equation}
    R_t(\alpha)=\mathcal{O}(\frac{1}{\kappa\sqrt{\tau}\Lambda})+\mathcal{O}(\sqrt{\tau})\,.
\end{equation}
Which of these two terms dominates is exactly given by the behaviour of $\tau\kappa\Lambda$, which is not constrained by the physical hydrodynamic behaviour of the system but by the precision with which we describe it. With our smoothness convention \eqref{smooth}, the term $\mathcal{O}(\sqrt{\tau})$ dominates. In either case, $R_t(\alpha)$ is subleading in \eqref{decompo} from both conditions of \eqref{hydro}. Summing over $q$ in \eqref{sum1} we obtain if $m(\alpha)\neq m(\beta)$
\begin{equation}
\begin{aligned}
    &G_t(\cg{\alpha},\cg{\beta})=\frac{n}{iL(\alpha-\beta)}\Big[e^{-it\alpha^2} \partial_{b(\cg{\alpha})-b(\cg{\beta}), \intg{2t\kappa\alpha}} -e^{-it\beta^2} \partial_{b(\cg{\alpha})-b(\cg{\beta}), \intg{2t\kappa\beta}} \Big]+\mathcal{O}\left(\frac{\kappa^2 \sqrt{\tau}}{\alpha-\beta} \right)\,,
\end{aligned}
\end{equation}
with the notation \eqref{derivdetal}. Let us now consider the case $\alpha=\beta$. In this case we have for $q\neq 0$
\begin{equation}
    \sum_{\lambda\in K_q}\varphi(\lambda,\cg{\alpha})^* \varphi(\lambda,\cg{\beta})e^{-it\lambda^2}=\frac{n}{L^2}e^{2i\pi (q+1/2)\frac{b(\cg{\alpha})-b(\cg{\beta})}{n}}4\sin^2(\tfrac{\pi (q+1/2)}{n}) \sum_{\lambda\in K_q}\frac{e^{-it\lambda^2}}{(\alpha-\lambda)^2}\,.
\end{equation}
Using
\begin{equation}\label{sum2}
    \sum_{m\in\mathbb{Z}}\frac{e^{ix(m+\theta)}}{(m+\theta)^2}=\left(\frac{\pi }{\sin\pi\theta}\right)^2 e^{2i\pi \theta \intg{\frac{x}{2\pi}+\frac{1}{2}}}+\frac{i\pi}{\sin \pi\theta}\left(x-2\pi \intg{\tfrac{x}{2\pi}+\tfrac{1}{2}}\right)e^{2i\pi\theta\left( \intg{\frac{x}{2\pi}}+\frac{1}{2}\right)}\,,
\end{equation}
we have
\begin{equation}
\begin{aligned}
     &\sum_{\lambda\in\bar{K}_q}\varphi(\lambda,\cg{\alpha})^* \varphi(\lambda,\cg{\beta})e^{-it\lambda^2}=e^{2i\pi (q+1/2)\frac{b(\cg{\alpha})-b(\cg{\beta})}{n}} \frac{e^{-it\alpha^2}}{n}\\
     &\times\left[e^{-\frac{2i\pi (q+1/2)}{n}\intg{2t\kappa\alpha+1/2}}-(2t\kappa\alpha-\intg{2t\kappa\alpha+1/2})(1-e^{-\frac{2i\pi}{n}(q+1/2)}) e^{-\frac{2i\pi (q+1/2)}{n}\intg{2t\kappa\alpha}} \right]\\
     &+\mathcal{O}(\sqrt{\tau})\,.
\end{aligned}
\end{equation}
Performing the sum over $q$ yields when $\alpha=\beta$
\begin{equation}
\begin{aligned}
&G_t(\cg{\alpha},\cg{\beta})=e^{-it\alpha^2}\delta_{b(\cg{\alpha})-b(\cg{\beta})-\intg{2t\kappa\alpha+1/2}}-e^{-it\alpha^2}(2t\kappa\alpha-\intg{2t\kappa\alpha+1/2})\partial_{b(\cg{\alpha})-b(\cg{\beta}),\intg{2t\kappa\alpha}}\\
&+\mathcal{O}(\kappa\sqrt{\tau})\,.
\end{aligned}
\end{equation}
These are the expressions given in the Lemma.
\end{proof}

\subsection{Hydrodynamic behaviour in absence of chemical potentials}
We now would like to study the behaviour of the system out of equilibrium in the hydrodynamic limit. In absence of chemical potentials, we have $\psi_t(\cg{\alpha},\cg{\beta})=G_t(\cg{\alpha},\cg{\beta})$. 

\subsubsection{Particle content \label{particlecontent}}
Let us first probe the particle content of box $b$ at time $t$. To that end, we consider the observable $\mathcal{O}_b(w)$ introduced in \eqref{obw} for a function $w(\beta)$. Its expectation value in an equilibrium state in box $b$ with root density $\rho$ is given by $\int w(\beta)\rho(\beta) \D{\beta}$, allowing one to probe the value of the root density $\rho$ in each box. With the expression \eqref{obwcumul}, Theorem \ref{theocumul} and the fact that $\psi_t(\cg{\alpha},\cg{\beta})=G_t(\cg{\alpha},\cg{\beta})$, we have at all times $t$
\begin{equation}
    \langle \mathcal{O}_b(w)\rangle_t=\sum_{q=1}^N\sum_{\substack{\cg{\beta}\in K_\otimes\\ b(\cg{\beta})=b}} w(\beta)|G_t(\cg{\alpha}_q,\cg{\beta})|^2\,,
\end{equation}
where we recall that $\pmb{\cg{\alpha}}=\{\cg{\alpha}_1,...,\cg{\alpha}_N\}$ is the initial wavelet state. In the hydrodynamic limit, this suggests to introduce for a fixed function $w(\beta)$
    \begin{equation}\label{palpha}
     P^{\rm hydro}_w(\cg{\alpha};0)  =\sum_{\substack{\cg{\beta}\in K_\otimes\\ b(\cg{\beta})=b}} w(\beta)|G_t^{\rm hydro}(\cg{\alpha},\cg{\beta})|^2\,,
    \end{equation}
for $\cg{\alpha}\in K_\otimes$, with $G_t^{\rm hydro}$ defined in Lemma \ref{hydrodynamicG}. Let us evaluate the hydrodynamic limit of this quantity. Importantly, we will assume that $w(\beta)$ varies only at scale much larger than $\kappa$ the typical spacing between energy levels in a box. This means we will assume
\begin{equation}
    \frac{w'(\beta)}{w(\beta)}=o(\kappa^{-1})\,.
\end{equation}
We note that this means the Fourier transform of $w(\beta)$ varies on a scale that is much smaller than the size of a box $1/\kappa$.

Using the expression for $G_t^{\rm hydro}$, we have
    \begin{equation}
    \begin{aligned}
       P^{\rm hydro}_w(\cg{\alpha};0)=&w(\alpha)\delta_{b(\cg{\alpha})-b-\intg{2t\kappa\alpha+1/2}}(1-2|2t\kappa\alpha-\intg{2t\kappa\alpha+1/2}|)\\
        &+w(\alpha)(2t\kappa\alpha-\intg{2t\kappa\alpha+1/2})^2 (\partial_{b(\cg{\alpha})-b, \intg{2t\kappa\alpha}})^2\\
        &+\frac{n^2}{L^2}\sum_{\substack{\beta \in K_{\rm box}\\ \beta\neq\alpha}} \frac{w(\beta)}{(\alpha-\beta)^2}|e^{-it\alpha^2} \partial_{b(\cg{\alpha})-b, \intg{2t\kappa\alpha}} -e^{-it\beta^2} \partial_{b(\cg{\alpha})-b, \intg{2t\kappa\beta}}|^2\,.
    \end{aligned}
    \end{equation}
We used that
\begin{equation}
    \tilde{\delta}_{j-\intg{x+1/2}}\partial_{j,\intg{x}}=\sign(x-\intg{x+1/2})\delta_{j-\intg{x+1/2}}\,.
\end{equation}
Writing $\alpha$ and $\beta$ in terms of $\kappa$ times integers, we have that only finite integer difference (but that can be arbitrarily large) will contribute to the sum in the hydrodynamic limit, because of the factor $\frac{1}{(\alpha-\beta)^2}$. But then we have $\intg{2t\kappa\alpha}=\intg{2t\kappa\beta}$, from assumption (ii) of \eqref{hydro}. Since $w'/w$ is of order $o(\kappa^{-1})$, approximating $w(\beta)$ by $w(\alpha)$ comes with an error $o(1)$ which goes to $0$ in the hydrodynamic limit. Hence

    \begin{equation}
    \begin{aligned}
        P^{\rm hydro}_w(\cg{\alpha};0)=&w(\alpha)\delta_{b(\cg{\alpha})-b-\intg{2t\kappa\alpha+1/2}}(1-2|2t\kappa\alpha-\intg{2t\kappa\alpha+1/2}|)\\
        &+w(\alpha)(2t\kappa\alpha-\intg{2t\kappa\alpha+1/2})^2 (\partial_{b(\cg{\alpha})-b, \intg{2t\kappa\alpha}})^2\\
        &+\frac{w(\alpha)(\partial_{b(\cg{\alpha})-b, \intg{2t\kappa\alpha}})^2}{4\pi^2}\sum_{m\neq 0} \frac{1}{m^2}|1-e^{4i\pi t\kappa\alpha m}|^2+o(1)\,.
    \end{aligned}
    \end{equation}
Using \eqref{sum2} in the limit $\theta\to 0$ we have
\begin{equation}
    \sum_{m\neq 0}\frac{e^{ixm}}{m^2}=\frac{\pi^2}{3}-\pi|x-2\pi \intg{\tfrac{x}{2\pi}+\tfrac{1}{2}}|+\frac{1}{2}(x-2\pi \intg{\tfrac{x}{2\pi}+\tfrac{1}{2}})^2\,.
\end{equation}
This yields
 \begin{equation}
    \begin{aligned}
        &P^{\rm hydro}_w(\cg{\alpha};0)=w(\alpha)\delta_{b(\cg{\alpha})-b-\intg{2t\kappa\alpha+1/2}}\\
        &-w(\alpha)|2t\kappa\alpha-\intg{2t\kappa\alpha+1/2}|(2\delta_{b(\cg{\alpha})-b-\intg{2t\kappa\alpha+1/2}}-\delta_{b(\cg{\alpha})-b-\intg{2t\kappa\alpha}}-\delta_{b(\cg{\alpha})-b-\intg{2t\kappa\alpha}-1})\,.
    \end{aligned}
    \end{equation}
    It can be rewritten as
    \begin{equation}\label{respalpha}
       P^{\rm hydro}_w(\cg{\alpha};0)=w(\alpha) (1-|b(\cg{\alpha})-b-2t\kappa\alpha|) \pmb{1}_{b(\cg{\alpha})-b-1<2t\kappa\alpha<b(\cg{\alpha})-b+1}\,.
    \end{equation}
    This contribution of a single initial particle $\cg{\alpha}$ to the root density is plotted in the left panel of Fig \ref{interference}. It can be exactly interpreted as the contribution of a "classical" particle initially located in box $b(\cg{\alpha})$ and moving from box to box with velocity $2\alpha$. This precisely corresponds to the interpretation of the GHD equations for the hardcore boson gas.

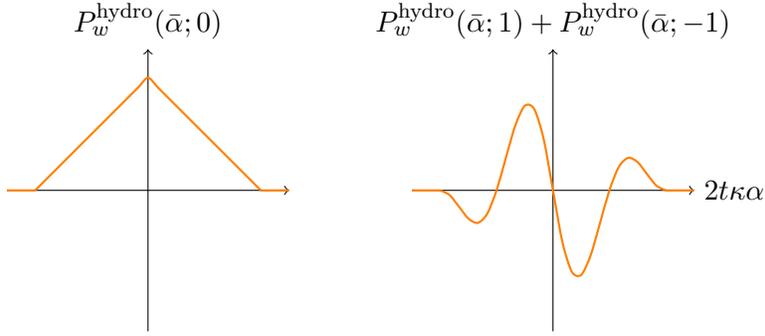
\begin{figure}[H]
\begin{center}
    \begin{tikzpicture}[scale=1.5]
  \draw[->] (-1.25, 0) -- (1.25, 0) node[right]{} ;
  \draw[->] (0, -1.25) -- (0, 1.25) node[above] {$P^{\rm hydro}_w(\cg{\alpha};0)$};
  \draw[domain=-1:1, smooth, variable=\x, orange, thick] plot ({\x}, {1-abs(\x)});
  \draw[domain=1:1.25, smooth, variable=\x, orange, thick] plot ({\x}, {0});
  \draw[domain=-1.25:-1., smooth, variable=\x, orange, thick] plot ({\x}, {0});
\end{tikzpicture}
\hspace{0.5cm}
    \begin{tikzpicture}[scale=1.5]
  \draw[->] (-1.25, 0) -- (1.25, 0) node[right] {$2t\kappa\alpha$};
  \draw[->] (0, -1.25) -- (0, 1.25) node[above] {$P^{\rm hydro}_w(\cg{\alpha};1)+P^{\rm hydro}_w(\cg{\alpha};-1)$};
  \draw[domain=-1:1, smooth, variable=\x, orange, thick] plot ({\x}, {-(1-abs(\x))*sin(90*2*2*\x)});
  \draw[domain=1:1.25, smooth, variable=\x, orange, thick] plot ({\x}, {0});
  \draw[domain=-1.25:-1., smooth, variable=\x, orange, thick] plot ({\x}, {0});
\end{tikzpicture}
\end{center}
\caption{$P^{\rm hydro}_w(\cg{\alpha};0)$ (\eqref{palpha}, left) and $P^{\rm hydro}_w(\cg{\alpha};1)+P^{\rm hydro}_w(\cg{\alpha};-1)$ (\eqref{palphadelta}, right) as a function of $2t\kappa\alpha$. The left panel represents the contribution of initial particle $\cg{\alpha}$ to the root density in the hydrodynamic limit, and the right panel to an ``anomalous" correlation function that vanishes at equilibrium. In presence of several initial particles, the function on the right will undergo destructive interference.}\label{interference}
\end{figure}
    
\subsubsection{Local relaxation: general remarks \label{generalremarks}}
We showed in Section \ref{particlecontent} that the particle content of box $b$ in the hydrodynamic limit exactly follows the GHD equations. The effective root density $\rho_{b,t}$ in box $b$ at time $t$ is then given by
\begin{equation}\label{effectiveroot}
    \rho_{b,t}(\cg{\beta})=\sum_{q=1}^N |\psi_t(\cg{\alpha}_q,\cg{\beta})|^2\,,
\end{equation}
where $\cg{\beta}\in K_\otimes$ with $b(\cg{\beta})=b$. We say it is ``effective" since it is only  defined by the fact that the expectation value of $\mathcal{O}_b(w)$ at time $t$ is $\sum_{\beta}\rho_b(\beta)w(\beta)$, as in an equilibrium state. But for it to really be the root density describing local observables in box $b$, one needs as well \emph{local relaxation}. Let us take the example of the connected two-point function of the density operator defined in \eqref{2pt}. Its expectation value is
\begin{equation}
\begin{aligned}
     \langle \mathcal{O}_b\rangle=\frac{1}{L^2}\sum_{h_1,h_2,p_1,p_2\in K_{\rm box}} \sum_{q_1,q_2=1}^N &e^{ix(p_1-h_1)+iy(p_2-h_2)}\\
    & \times \det \left(\begin{matrix}
    \psi_t(\cg{\alpha}_{q_1},\cg{h}_1)\psi_t(\cg{\alpha}_{q_1},\cg{p}_1)^*& \psi_t(\cg{\alpha}_{q_1},\cg{h}_1)\psi_t(\cg{\alpha}_{q_2},\cg{p}_1)^*\\
    \psi_t(\cg{\alpha}_{q_2},\cg{h}_2)\psi_t(\cg{\alpha}_{q_1},\cg{p}_2)^*& \psi_t(\cg{\alpha}_{q_2},\cg{h}_2)\psi_t(\cg{\alpha}_{q_2},\cg{p}_2)^*
    \end{matrix}\right)\,.
\end{aligned}
\end{equation}
This can be written as
\begin{equation}\label{correspondto}
\begin{aligned}
     \langle \mathcal{O}_b\rangle=&\langle \psi^\dagger(x)\psi(x)\rangle_t \langle \psi^\dagger(y)\psi(y)\rangle_t\\
    &-\frac{1}{L^2} \sum_{q_1,q_2=1}^N\sum_{\delta,\delta'\in K}  e^{iy\delta+ix\delta'} P_w(\cg{\alpha}_{q_1};\delta) P_{w^*}(\cg{\alpha}_{q_2};\delta')\,,
\end{aligned}
\end{equation}
with $w(\beta)=e^{i\beta(y-x)}$ and where  we introduced
\begin{equation}
    P_w(\cg{\alpha};\delta)=\sum_{{\substack{\cg{\beta}\in K_\otimes\\b(\cg{\beta})=b}}}w(\beta)G^{\rm hydro}_t(\cg{\alpha},\cg{\beta})G^{\rm hydro}_t(\cg{\alpha},\cg{\beta}+2\pi\kappa\delta)^*\,,
\end{equation}
where $\cg{\beta}+2\pi\kappa\delta\in K_\otimes$ denotes the wavelet in same box as $\cg{\beta}$ but with momentum $\beta+2\pi\kappa\delta$. In an equilibrium state $|\pmb{\beta}\rangle$, this expectation value would be
\begin{equation}
     \langle\pmb{\beta}| \mathcal{O}_b|\pmb{\beta}\rangle=\langle \psi^\dagger(x)\psi(x)\rangle_t \langle \psi^\dagger(y)\psi(y)\rangle_t-\frac{1}{L^2}\sum_{\beta,\beta'\in\pmb{\beta}}w(\beta)w^*(\beta')\,.
\end{equation}
Hence we see that this formula \emph{does not} correspond to   \eqref{correspondto} with the effective root density \eqref{effectiveroot}. The effective root density would only capture the terms with $\delta=\delta'=0$. To show local relaxation to an equilibrium state in box $b$ in the hydrodynamic limit, one thus needs to show that the terms for which $\delta\neq 0$ or $\delta'\neq 0$ vanish in the hydrodynamic limit.

\subsubsection{Local relaxation: simplest case \label{simple}}
Let us define the hydrodynamic limit of $P_w(\cg{\alpha};\delta)$
\begin{equation}\label{palphadelta}
    P^{\rm hydro}_w(\cg{\alpha};\delta)=\sum_{{\substack{\cg{\beta}\in K_\otimes\\b(\cg{\beta})=b}}}w(\beta)G^{\rm hydro}_t(\cg{\alpha},\cg{\beta})G^{\rm hydro}_t(\cg{\alpha},\cg{\beta}+2\pi\kappa\delta)^*\,,
\end{equation}
for a fixed function $w(\beta)$ that varies on a scale larger than $1/\kappa$. We note that for the example of the density $2$-point function of Section \ref{generalremarks} where we have $w(\beta)=e^{i\beta(y-x)}$, this condition means that the distance separation $y-x$ is much smaller than the size of the box $1/\kappa$. 

We write
\begin{equation}\label{pdalpha}
\begin{aligned}
 P^{\rm hydro}_w(\cg{\alpha};\delta)=&w(\alpha)G_t(\cg{\alpha},\cg{\alpha})G_t(\cg{\alpha},\cg{\alpha}+2\pi\kappa\delta)^*+w(\alpha-2\pi\kappa\delta)G_t(\cg{\alpha},\cg{\alpha}-2\pi\kappa\delta)G_t(\cg{\alpha},\cg{\alpha})^* \\
&+\sum_{\substack{\cg{\beta}\in K_\otimes\\ b(\cg{\beta})=b\\ \beta\neq \alpha,\alpha-2\pi\kappa\delta}} w(\beta)G_t(\cg{\alpha},\cg{\beta})G_t(\cg{\alpha},\cg{\beta}+2\pi\kappa\delta)^*
\end{aligned}
\end{equation}
and use the expression for $G_t^{\rm hydro}(\cg{\alpha},\cg{\beta})$ in Lemma \ref{hydrodynamicG}. Let us first focus on the sum, that we denote $S$. We have
\begin{equation}
\begin{aligned}
    S=\sum_{\substack{\beta\in K_{\rm box}\\ \beta\neq \alpha,\alpha-2\pi\kappa\delta}} &\frac{w(\beta)\kappa^2}{(\alpha-\beta)(\alpha-\beta-2\pi\kappa\delta)}\Big[e^{-it\alpha^2} \partial_{b(\cg{\alpha})-b, \lfloor 2t\kappa\alpha\rfloor } -e^{-it\beta^2} \partial_{b(\cg{\alpha})-b, \lfloor 2t\kappa\beta\rfloor} \Big]\\
    &\times\Big[e^{it\alpha^2} \partial_{b(\cg{\alpha})-b, \lfloor 2t\kappa\alpha\rfloor } -e^{it(\beta+2\pi\kappa\delta)^2} \partial_{b(\cg{\alpha})-b, \lfloor 2t\kappa(\beta+2\pi\kappa\delta)\rfloor} \Big]\,.
\end{aligned}
\end{equation}
The terms with $\alpha-\beta$ much larger than $\kappa$ go to zero, whereas the terms with $\alpha-\beta=\mathcal{O}(\kappa)$ are of order $\mathcal{O}(1)$. Hence only the terms for which $\alpha-\beta=\mathcal{O}(\kappa)$ contribute. But then, we have $\intg{2t\kappa\alpha}=\intg{2t\kappa\beta}=\intg{2t\kappa(\beta+2\pi\delta\kappa)}$ because $t\kappa^2\to 0$ from \eqref{smooth}, and $w(\beta)=w(\alpha)+o(1)$. Hence
\begin{equation}
    S=w(\alpha)\kappa^2 \partial_{b(\cg{\alpha})-b,\intg{2t\kappa\alpha}}^2 \sum_{\beta\in K_{\rm box}}\frac{(1-e^{it(\alpha^2-\beta^2)})(1-e^{it((\beta+2\pi\kappa\delta)^2-\alpha^2)}}{(\alpha-\beta)(\alpha-\beta-2\pi\kappa\delta)}+o(1)\,.
\end{equation}
Writing $\beta$ and $\alpha$ in terms of integers, and using $\kappa^2 t\to 0$ from \eqref{hydro}, this yields
\begin{equation}
    S=\frac{w(\alpha)\partial_{b(\cg{\alpha})-b,\intg{2t\kappa\alpha}}^2}{4\pi^2}\sum_{m\in\mathbb{Z}}\frac{2-e^{4i\pi\kappa t \alpha(m+\delta)}-e^{4i\pi\kappa t \alpha m}}{m(m+\delta)}+o(1)\,.
\end{equation}
Using \eqref{sum1} in the limit $\theta\to 0$ we get
\begin{equation}
    \sum_{m\neq 0}\frac{e^{ixm}}{m}=i\pi(1-|x|)\sign(x)\,,\qquad -\pi\leq x<\pi\,.
\end{equation}
This yields
\begin{equation}
    S=o(1)\,.
\end{equation}
Coming back to \eqref{pdalpha} and using \eqref{gequal} we obtain
\begin{equation}
\begin{aligned}
 P^{\rm hydro}_w(\cg{\alpha};\delta)=&\frac{w(\alpha)}{i\pi\delta}(\delta_{b(\cg{\alpha})-b-\intg{2t\kappa\alpha+1/2}}-|2t\kappa\alpha-\intg{2t\kappa\alpha+1/2}|(\delta_{b(\cg{\alpha})-b-\intg{2t\kappa\alpha}}+\delta_{b(\cg{\alpha})-b-\intg{2t\kappa\alpha}-1}))\\
&\times(1-e^{4i\pi t \kappa\delta\alpha})\sign(2t\kappa\alpha-\intg{2t\kappa\alpha+1/2})\,.
\end{aligned}
\end{equation}
This contribution of a single initial particle $\cg{\alpha}$ is plotted in the right panel of Fig \ref{interference}. Contrary to $ P^{\rm hydro}_w(\cg{\alpha};0)$, these contributions $ P^{\rm hydro}_w(\cg{\alpha};\delta)$ vanish when integrated over $\cg{\alpha}$. Hence when several particles are included in the initial state with close momenta, these contributions vanish.

\subsubsection{Local relaxation: summary}
Let us summarize our findings. In Section \ref{particlecontent} we showed that in the hydrodynamic limit defined in \eqref{hydro}, the effective root density (and so the expectation value of any conserved charges) follows from the GHD equations \emph{particle by particle}. This means that as long as one is interested in conserved charges, \emph{no large number of particles} is required for the GHD description to be valid. 

Now, the usual meaning of hydrodynamics is that a state can be considered to be in an equilibrium state as far as expectation values of local observables are concerned. This means more than just having the right expectation values of conserved charges, as explained in Section \ref{generalremarks}, since these do not fix all correlations functions in an out-of-equilibrium state (contrary to equilibrium). In Section \ref{simple}, we showed that this local relaxation is not achieved for single particles, since there are some terms that do not vanish in the hydrodynamics limit \eqref{hydro} whereas they would vanish in an equilibrium state. However, these terms have zero integral and so will undergo destructive interference when several particles are included in the initial state. Hence, to summarize, local relaxation requires a large number of particles to occur, but expectation values of conserved charges \emph{do not} require a large number of particles to be described by GHD in the limit \eqref{hydro}.

\section{Algorithm for simulation of hardcore bosons \label{algosec}}
\subsection{The algorithm}
Theorem \ref{eith} allows for an efficient algorithm to simulate hardcorse bosons in 1D. According to this result, the function $\psi_t(\cg{\alpha},\cg{\beta})$ fully characterizes the state during the out-of-equilibrium protocol described in Section \ref{problemsetting}. Expectation values of local observables can be expressed directly in terms of this function, according to Theorem \ref{theocumul}. This function $\psi_t(\cg{\alpha},\cg{\beta})$ represents the state in a way that is mixed between real space (through the dependence in the box indices $b(\cg{\alpha}),b(\cg{\beta})\in\{1,...,n\}$ of $\cg{\alpha},\cg{\beta}\in K_\otimes$) and momentum space (through the dependence in the momenta $\alpha,\beta\in K_{\rm box}$). This mixed representation is key to the efficiency of the algorithm. Its resolution in space is well adapted to inhomogeneous settings and the computation of local observables, while its resolution in momentum is well adapted to compute the time evolution of the state efficiently and with good precision.\\

In practice, we consider an even number of boxes $n$ and define $K_\otimes^{(P)}$ for $P$ an odd integer a truncated version of $K_\otimes$
\begin{equation}
    K_\otimes^{(P)}= \left\{\left(\frac{2\pi nm}{L},b\right),\quad m\in\{-(P-1)/2,...,(P-1)/2\}\,,\quad b\in\{1,...,n\}\right\}\,,
\end{equation}
and $K^{(P)}$ a truncated version of $K$
\begin{equation}
   K^{(P)}= \left\{\frac{2\pi (m+1/2)}{L},\quad m\in\{-nP/2,...,nP/2-1\}\right\}\,,
\end{equation}
and see $\psi_t(\cg{\alpha},\cg{\beta})$ at fixed $t$ as a $nP\times nP$ matrix. For a given time step $\delta t$, we compute its time evolution iteratively through
\begin{equation}\label{ierpsi}
    \psi_{t+\delta t}(\cg{\alpha},\cg{\beta})=\sum_{\cg{\gamma}\in K^{(P)}_\otimes} \mathcal{U}(\cg{\alpha},\cg{\gamma})\psi_t(\cg{\gamma},\cg{\beta})\,,
\end{equation}
with the unitary matrix
\begin{equation}
    \mathcal{U}(\cg{\alpha},\cg{\beta})=\sum_{\lambda\in K^{(P)}}\varphi^{(P)}(\lambda,\cg{\alpha})^*\varphi^{(P)}(\lambda,\cg{\beta})e^{-i\delta t\lambda^2}e^{-i\delta t v_{b(\cg{\beta})}}\,.
\end{equation}
Here, we defined $\varphi^{(P)}(\lambda,\cg{\alpha})$ as the $nP\times nP$ unitary matrix obtained by applying a QR decomposition to the  $nP\times nP$ matrix $(\varphi(\lambda,\cg{\alpha}))_{\lambda\in K^{(P)},\cg{\alpha}\in K_\otimes^{(P)}}$. The iteration \eqref{ierpsi} can be implemented as a simple matrix multiplication. The exact time evolution is recovered in the limits $\delta t\to 0$ and $P\to \infty$.\\

Some comments on the QR decomposition are in order. Firstly (i), the unitary matrix $\varphi^{(P)}(\lambda,\cg{\alpha})$ obtained from a QR decomposition is unique only up to a right multiplication by a diagonal matrix with modulus $1$ coefficients. But this corresponds to conjugating $\psi_t$ with this diagonal unitary matrix, which has no effects on the expectation values of observables. Secondly (ii), the QR decomposition while preserving the unitarity of $\mathcal{U}$ breaks parity symmetry, namely the resulting $\mathcal{U}$ is not invariant under simultaneous reversing of the box indices $b\to n+1-b$ and momenta $\alpha,\beta\to-\alpha,-\beta$, whereas it would be so without the QR decomposition. This will introduce a slight left/right asymmetry in the results even if the initial condition is parity symmetric. This symmetry is however restored in the limit $P\to\infty$.

\subsection{Quantum Newton cradle setup}
\subsubsection{Protocol}
As an example of application of the algorithm, we consider a quantum Newton cradle setup \cite{kinoshita2006quantum,van2016separation}. At time $t<0$, the system of $N$ particles is initialized in the ground state of the Tonks-Girardeau model with a harmonic potential
\begin{equation}\label{harmonic}
    V(x)=\omega^2 x^2\,.
\end{equation}
Then at time $t=0$ a "Bragg pulse" is applied to the system, that is a potential given by
\begin{equation}\label{bragg}
    V_{\rm pulse}(x)=A\cos(qx)\,,
\end{equation}
with a large space frequency $q$ and large amplitude $A$, for a short time window $\Delta t_{\rm pulse}$. Then the system is time evolved with the original potential \eqref{harmonic}. 

\subsubsection{Numerical implementation \label{braggsection}}
In our setting, the initial state is necessarily a tensor product of energy eigenstates in each of the boxes, which cannot encode the ground state of the Tonks-Girardeau model with the potential \eqref{harmonic}. We thus have to include a first initial stage to prepare this ground state, before applying the Bragg pulse. We proceed as follows. We initialize the system with one particle in the zero momentum space of each of the $n/4$ boxes from box $3n/8$ to $5n/8$. Numerically, we take $n=256$ boxes (so $N=64$ particles) and set the system size to $L=32$. Then from $t=0$ to $t=t_0$ we slowly change the potential from an infinite (or large) square well to the potential \eqref{harmonic}. At time $t=0$, the wave function relaxes quickly to the ground state of the square well potential, which then according to the adiabatic theorem remains in the ground state of the time-dependent potential, provided the potential changes slowly enough. Specifically, we apply the following potential for $0<t<t_0$
\begin{equation}
    V_t(x)=\omega^2 x^{2t/t_0} |8x|^{10(t_0-t)}\,,
\end{equation}
with $x=\tfrac{b-1/2}{n}-\tfrac{1}{2}$ where $b\in\{1,...,n\}$ is the box index, and with the numerical values $t_0=2$ and $\omega^2=16000$. Then, at time $t_0$ and onward, we fix the potential to be \eqref{harmonic}. 

To modelize the Bragg pulse, we consider the instantaneous limit with $\Delta t_{\rm pulse}\to 0$ and $A\to\infty$ at fixed $A\Delta t_{\rm pulse}$, as in \cite{van2016separation}. Then the effect of the pulse is at $t=t_0$ to apply on the wave function the operator $e^{-i\Delta t_{\rm pulse} \int \D{x} V_{\rm pulse}(x) \phi^\dagger(x)\phi(x)\D{x}}$. Since our setup requires the potential to be constant in each box, we approximate $V_{\rm pulse}$ to be constant equal to $A$ on even boxes, and to $-A$ on odd boxes. Numerically, we take $A\Delta t_{\rm pulse}=\pi/2$. 

Finally, for comparison purposes, we carry out the same simulation with a quartic potential $V_{\rm quartic}(x)\propto x^4$, with a proportionality constant chosen such that $V_{\rm quartic}(1/6)=V(1/6)$. 

We present the numerical results in Fig \ref{braggpulse} and \ref{braggpulse2}. In Fig \ref{braggpulse} we measure the number of particles per unit length at the middle of the system, averaging over $8$ boxes, and plot it as a function of time. We observe a large number of oscillations with a large amplitude for the harmonic potential. With the quartic potential, we see that the oscillations are immediately damped.

In Fig \ref{braggpulse2} we present the entire profile of expectation value of the particle number as a function of the box index $b$, at different times. We display as well the mode occupation number for the $P$ box momenta $\beta\in \{\tfrac{2\pi n p}{L},p=-(P-1)/2,...,(P-1)/2\}$, averaged over all the $n$ boxes.

\begin{figure}[H]
\begin{center}

\includegraphics[scale=1]{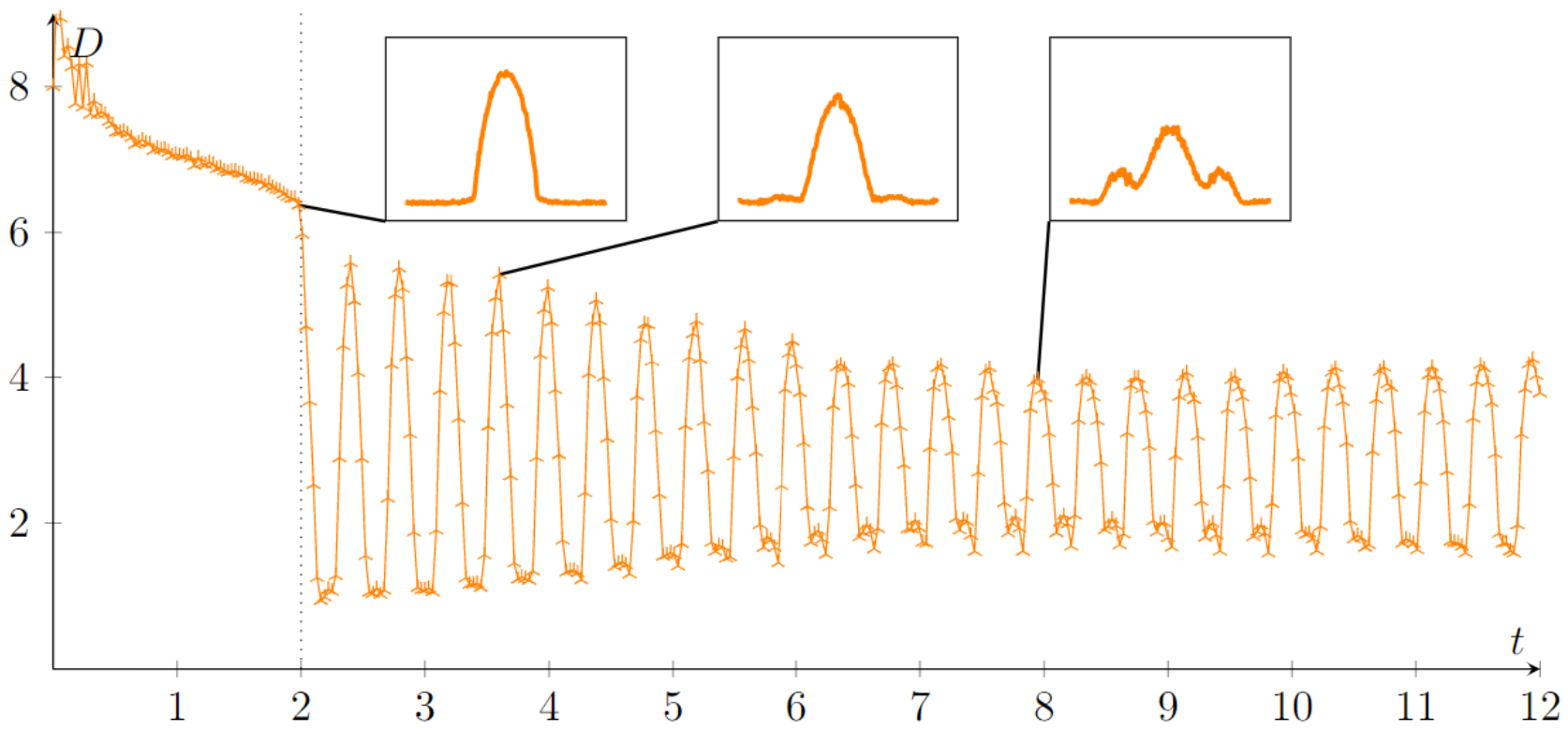}
\includegraphics[scale=1]{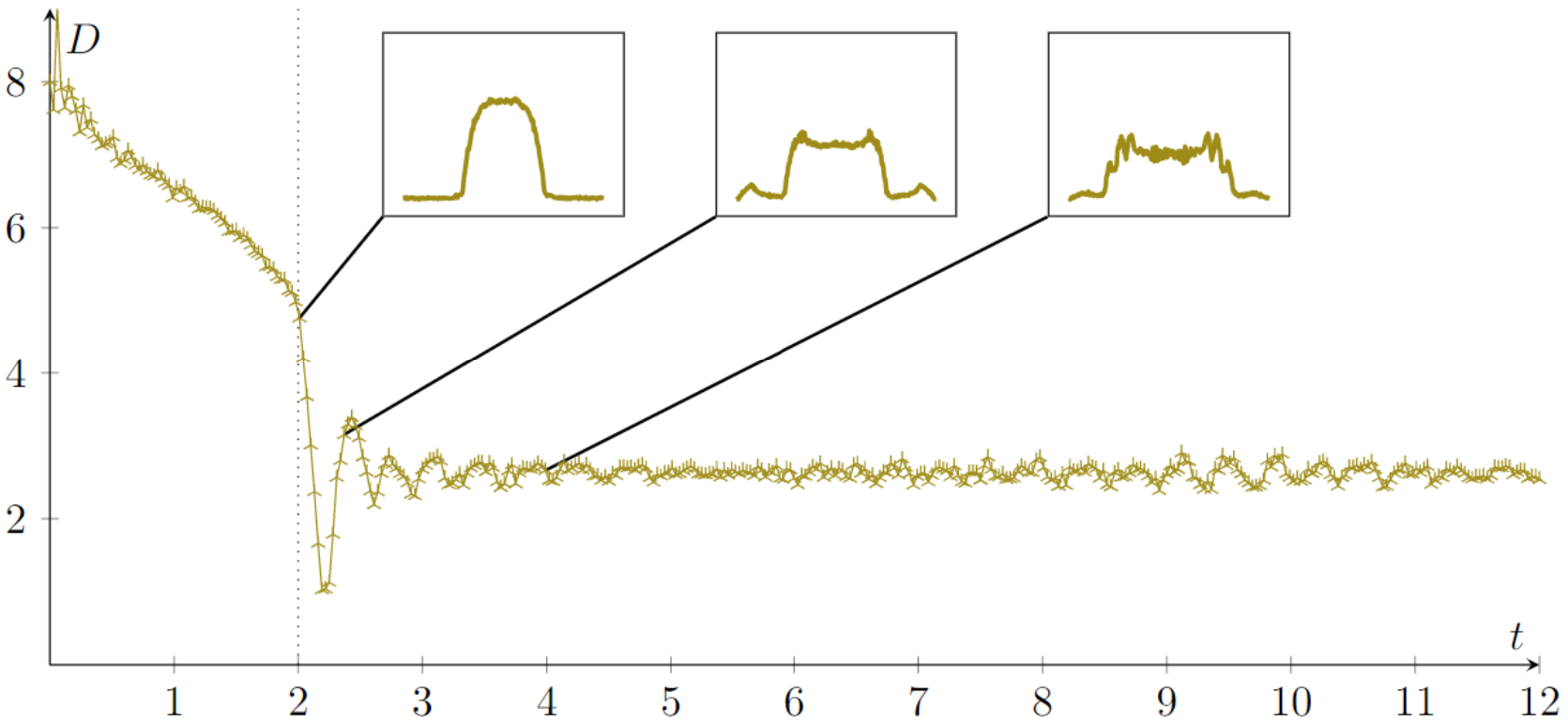}

\caption{\emph{Quantum Newton cradle setup}. Number of particles per unit length at the middle of the system as a function of time, for the setup described in Section \ref{braggsection}, with the quadratic potential ({\color{orange}orange}, top panel) and the quartic potential ({\color{olive}green}, bottom panel). The insets show the profile of the expectation value of the particle number in each box, as a function of the box index, at times indicated by the black line. The algorithm parameters are $\delta t=0.0002$ and $P=15$. The dotted vertical lines indicate the times at which the Bragg pulse is applied.} 
\label {braggpulse}
\end{center}
\end {figure}

\begin{figure}[H]
\begin{center}

\includegraphics[scale=1]{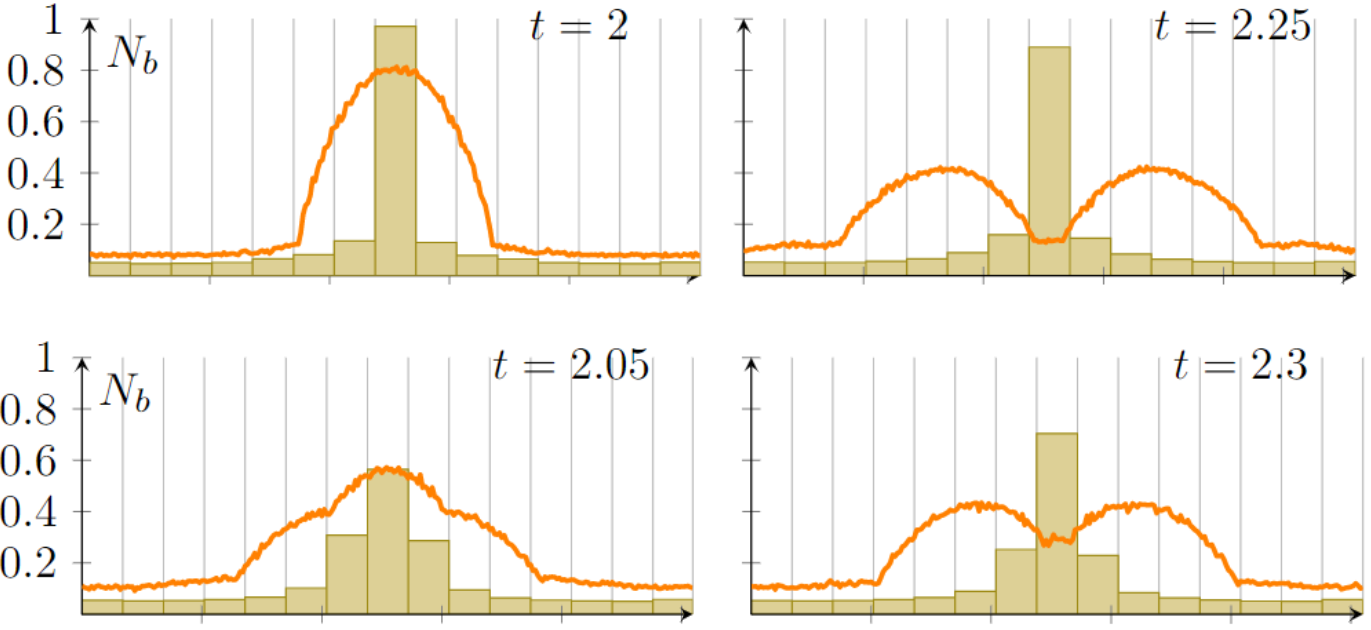}
\includegraphics[scale=1]{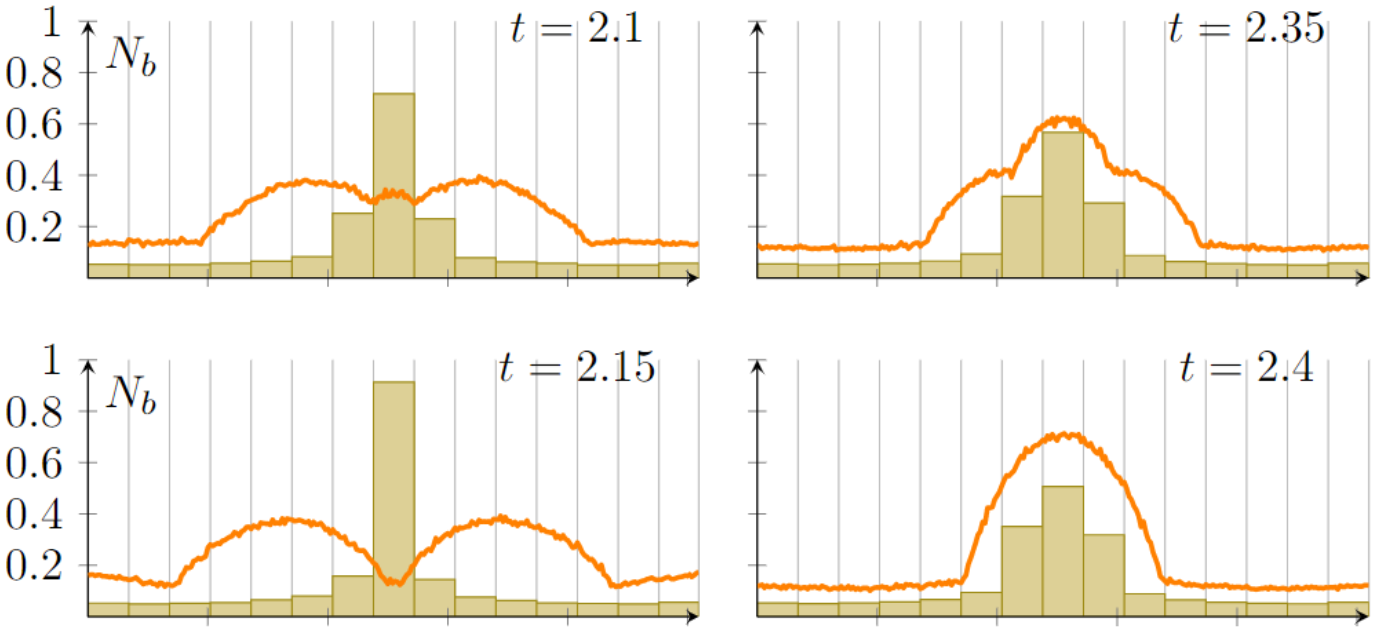}
\includegraphics[scale=1]{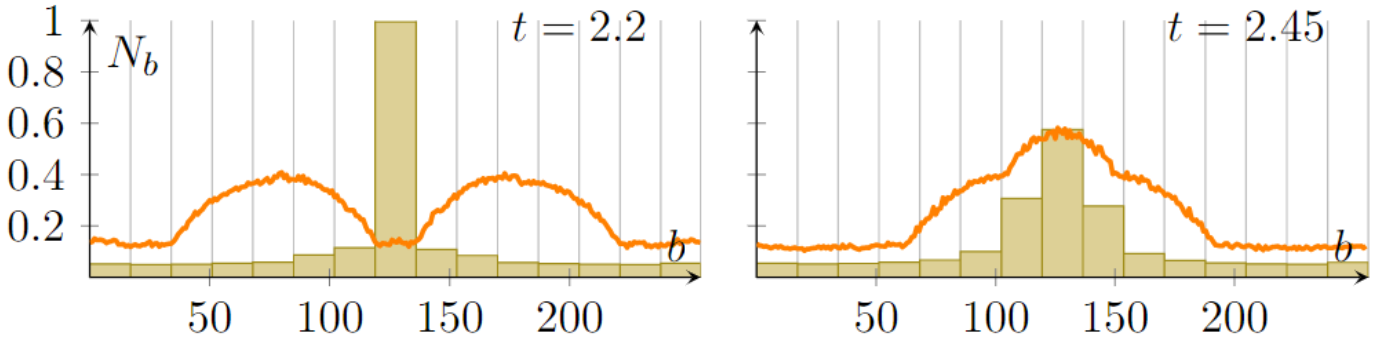}

\caption{\emph{Quantum Newton cradle setup}. Setup described in Section \ref{braggsection} with the harmonic potential, at different times $t$. Orange curve: expectation value of the particle number $N_b$ in each box, as a function of the box index $b$. Green bars: amplitude of each of the $P$ box mode occupation number $\alpha\in K^{(P)}_{\otimes}$, averaged over all the boxes, displayed in increasing order from left to right. For better readability, the vertical scale of the green bars is arbitrary but consistent for different times. The algorithm parameters are $\delta t=0.0001$ and $P=15$. The first plot at $t=2$ is just before the Bragg pulse, all the others correspond to times after the Bragg pulse.} 
\label {braggpulse2}
\end{center}
\end {figure}

\subsection{Convergence towards hydrodynamics \label{hydrodconvsec}}
We now study the convergence of the numerics with the hydrodynamic limit \eqref{hydro} in the quantum Newton cradle setup described in Section \ref{braggsection}. We implement the hydrodynamic limit \eqref{hydro} the following way. If the system is divided into $n$ boxes, we define the rescaled time
\begin{equation}
    \tau=\frac{256 t}{n}-\tau_0\,,
\end{equation}
with the rescaled Bragg pulse instant $\tau_0= \frac{512}{n}$, and the rescaled potential amplitude
\begin{equation}
    \omega^2=\frac{125}{512} n^2\,.
\end{equation}
The numbers are chosen so as to recover those of Section \ref{braggsection} for the value $n=256$. The initial state is always chosen to be one particle in the zero momentum mode of the $n/4$ boxes at the middle of the system, as in Section \ref{braggsection} (there are thus $N=n/4$ particles). The number of boxes $n$ plays the role here of the parameter to scale to reach the hydrodynamic limit \eqref{hydro}, which is obtained in the limit $n\to\infty$. We note that this exactly corresponds to a short-time, high-density limit mentioned below \eqref{hydro}. We measure the expectation value of the particle number at the middle of the system $N_{n/2}$. We show in Fig \ref{hydroconver} the results of the numerics.

\begin{figure}[H]
\begin{center}

\includegraphics[scale=1]{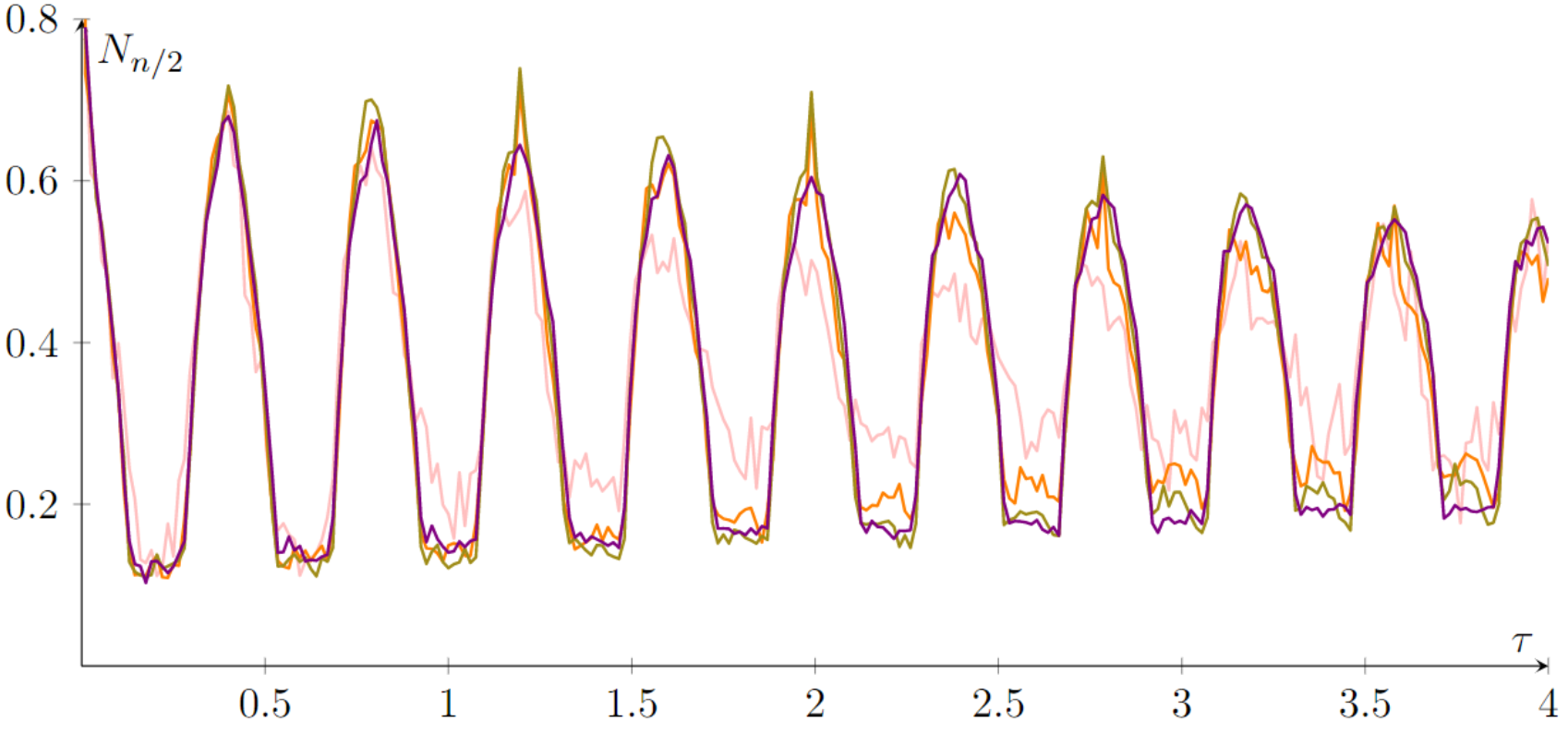}

\caption{\emph{Convergence towards the hydrodynamic limit}. Expectation value of the particle number at the middle of the system $N_{n/2}$ as a function of the rescaled time $\tau$, for the setup described in Section \ref{hydrodconvsec}, with $N=8$ particles ({\color{pink}pink}), $N=32$ particles ({\color{orange}orange}), $N=64$ particles ({\color{olive}green}), $N=128$ particles ({\color{violet}purple}). The algorithm parameters are $\delta t=0.0256/n$ and $P=7$.} 
\label {hydroconver}
\end{center}
\end {figure}

\subsection{Quench from double well to single well potential: convergence with \texorpdfstring{$\delta t$}{Lg} and \texorpdfstring{$P$}{Lg} \label{sec:conv}}
We now present numerical results of the algorithm in a setup analogous to the quantum Newton cradle setup, which is a quench from a double well potential to a single well potential. At time $t=0$, we take the initial wavelet state to be a tensor product of lowest energy eigenstates with $m$ particles in each of the leftmost $n/4$ boxes and rightmost $n/4$ boxes (the total state has thus $nm/2$ particles). We choose this particular initial setup only to ease the first relaxation stage with the double well potential. Then for time $0<t<t_0$ we impose the following potential
\begin{equation}
   V^{(1)}(x)=4\omega^2(-x^2+4x^4)\,, %v_b^{(1)}=4\left( -(\tfrac{b}{n})^2+4(\tfrac{b}{n})^4\right)\omega^2\,,
\end{equation}
with $x=\tfrac{b-1/2}{n}-\tfrac{1}{2}$ where $b\in\{1,...,n\}$ is the box index and with $\omega^2$ an overall scale. Then for time $t>t_0$ we time evolve the system with the single potential well
\begin{equation}
   V^{(2)}(x)=\omega^2 x^2\,.% v_b^{(2)}=(\tfrac{b}{n})^2\omega^2\,.
\end{equation}
We present numerical tests of the convergence speed of the algorithm with $\delta t$ and $P$. To fix values, with set $L=32$ the size of the system, $n=256$ the number of boxes, $m=1$ the number of particles in each occupied box initially (so $N=128$ particles in total), $V=8000$ the strength of the potential, and $t_0=2$ the duration of the first relaxation stage, for a total running time $t_{\rm fin}=16$. Then we measure $N_{b}$ the number of particles in box $b$ as a function of time $t$, for different algorithm parameters $\delta t$ and $P$. The results are plotted in Fig \ref{convergence} and \ref{convergence2}. We see that for these values of physical parameters, the algorithm parameters $\delta t=0.0005$ and $P=15$ show almost converged curves. The simulation with these parameters takes just a few minutes on a laptop.

\begin{figure}[H]
\begin{center}

\includegraphics[scale=1]{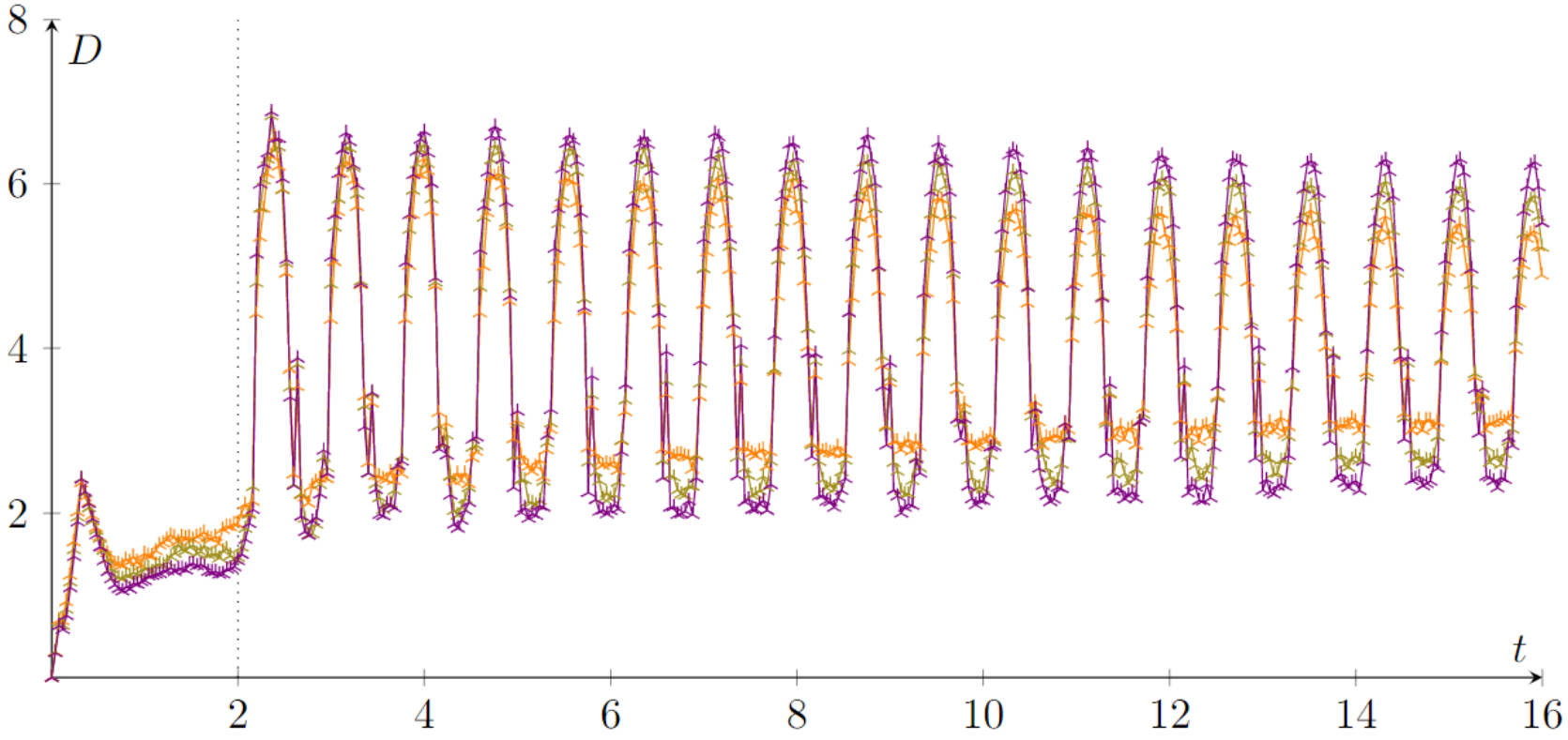}

\caption{\emph{Influence of the time step $\delta t$}. Number of particles per unit length at the middle of the system as a function of time, for the physical parameters quoted in Section \ref{sec:conv}, with $\delta t=0.005$ ({\color{orange}orange}), $\delta t=0.002$ ({\color{olive}green}), $\delta t=0.0005$ ({\color{violet}purple}), all for $P=15$. The dotted vertical lines indicate the moment of the quench from the double well to the single well.} 
\label {convergence}
\end{center}
\end {figure}

\begin{figure}[H]
\begin{center}

\includegraphics[scale=1]{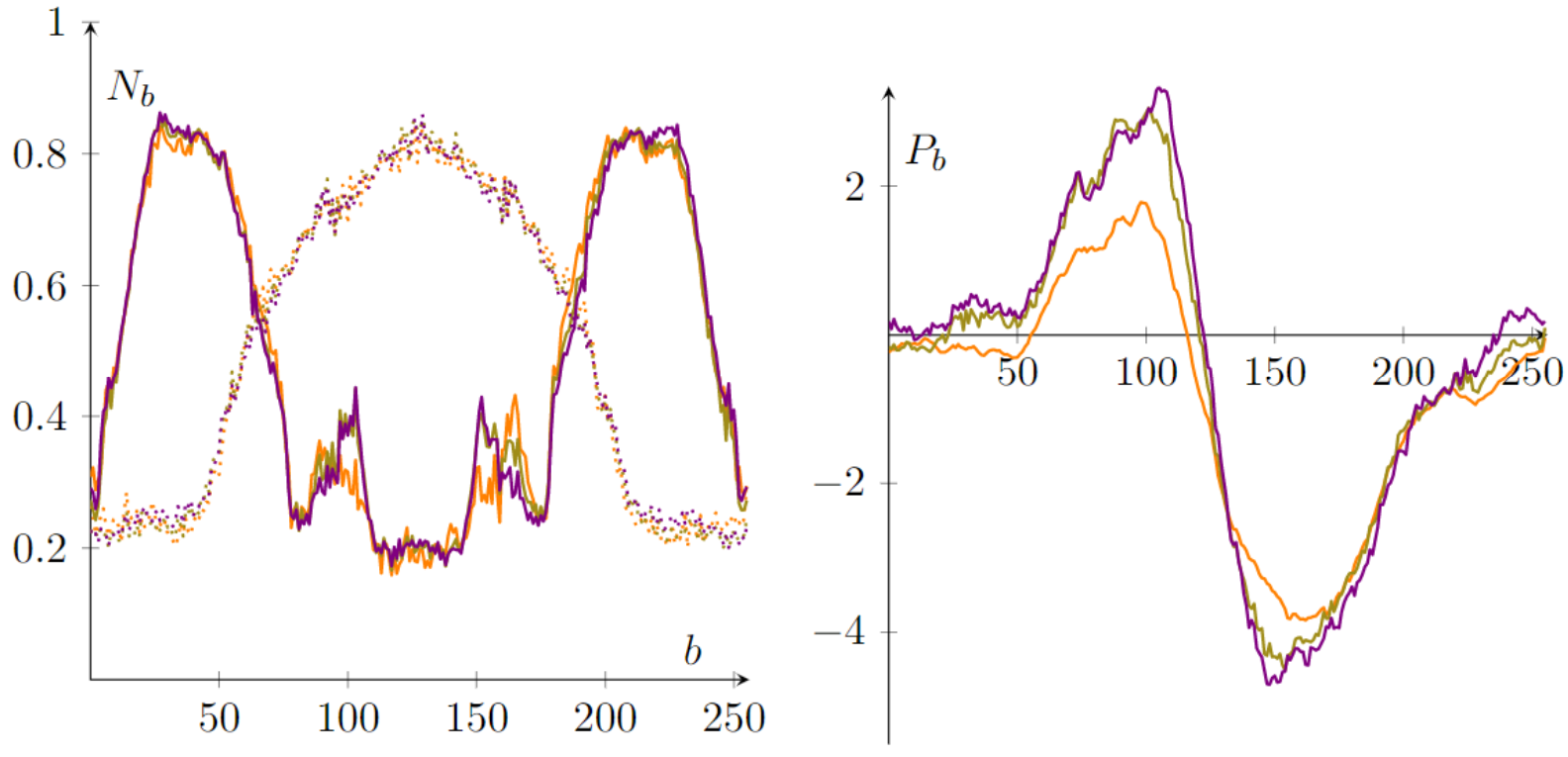}

\caption{\emph{Influence of truncation parameter $P$}. Expectation value of the number of particles (left panel) and of the momentum (right panel) in each box, as a function of the box index $b$, for the physical parameters quoted in Section \ref{sec:conv}. These are snapshots at times $t=2.75$ (continuous line) and $t=3.15$ (dotted) for the left panel, and at time $t=3$ for the right panel. The algorithm parameters are $\delta t=0.0005$ for all curves, with $P=3$ ({\color{orange}orange}), $P=7$ ({\color{olive}green}), $P=15$ ({\color{violet}purple}). The momentum curve is averaged over a time window of size $0.03$ to remove very fast oscillations.} 
\label {convergence2}
\end{center}
\end {figure}

\section{Summary and discussion}
In this paper we studied the emergence of hydrodynamics in a 1D gas of hardcore bosons in inhomogeneous potentials from first principles. We provided an exact implementation of the ubiquitous ``fluid cell" view of hydrodynamics. Specifically, we divided the system into $n$ different boxes and decomposed the out-of-equilibrium wave function of the system in the basis given by tensor products of energy eigenstates on each of the boxes. Such representation is mixed between real space and momentum space, and is composed of small plane waves localized in each box, called ``wavelets". Contrary to a hydrodynamic or coarse-grained description, there is \emph{no} removal of degrees of freedom in this wavelet representation. Because hardcore bosons can be mapped to free fermions, the dynamics in this wavelet basis can be exactly computed using form factor expansions. 

The output of this work is twofold. Firstly, we provide a derivation of (generalized) hydrodynamics for hardcore bosons from first principles, without the assumption of local equilibrium which was systematically required previously. In particular we emphasize that hydrodynamics emerges in a short-time, high-momentum limit. In this limit, the expectation values of conserved charges \emph{do not} require a large number of particles in the initial state to be described by GHD equations. However, local relaxation (that is required to deduce all local correlations from the conserved charges) is obtained only when including a large number of particles to produce a destructive interference of ``non-relaxed terms".  Regarding the cold-atom experiments applications of GHD, this seems more realistic than the Euler limit of large time and large space.  Secondly, we show that the decomposition of the system in this wavelet basis provides an efficient numerical algorithm to simulate the dynamics of hardcore bosons with inhomogeneous potentials. It is indeed a  representation that is mixed in real space (through the division into boxes) and momentum space (through the decomposition into energy eigenstates in each box) that can be implemented efficiently and that gives access to several observables for a large number of particles and for a relatively long time, including the single-particle Green's function, otherwise difficult to compute. 

The results presented in this paper suggest a number of directions. A natural direction would be to generalize this approach to finite coupling $c$ between the bosons. The strategy would be to perform a strong coupling $1/c$ expansion, which has proven efficient before \cite{granet2020systematic,granet2022duality,bertini2022bbgky}. At large finite coupling the bosons can be mapped to fermions with a weak coupling, allowing for perturbative expansions \cite{cheon1999fermion,brand2005dynamic,granet2022duality,granet2022regularization}. Another approach could be to generalize the results of geometric quenches to work directly at finite coupling $c$ \cite{mossel2010geometric}. Other possible directions include using the algorithm we introduced to simulate more closely cold-atom experiments, for example by including atom losses which play a significant role \cite{bouchoule2020effect,alba2022noninteracting,hutsalyuk2021integrability,rosso2022one,bouchoule2021losses}.

\vspace{0.5cm}

{\em \bf Acknowledgments:}
We thank Fabian Essler and Tony Jin for comments on the draft. This work was supported by the Kadanoff Center for Theoretical Physics at University of Chicago, and by the Simons Collaboration on Ultra-Quantum Matter.

%\bibliography{biblio}

\printbibliography

\end{document}